\documentclass[runningheads]{llncs}
\usepackage{url, hyperref}

\usepackage{amsmath,amssymb} 

\usepackage{program}
\usepackage{centernot}
\usepackage{graphicx}

\usepackage{comment}
\usepackage{xparse}
\usepackage{stmaryrd} 
\usepackage{here}
\usepackage{mleftright}
\usepackage{xcolor}
\usepackage{multicol}
\usepackage{multirow}
\usepackage{proof}
\usepackage[scaled=0.8]{beramono}
\usepackage[T1]{fontenc}
\usepackage[all]{xy}
\usepackage{pifont}
\usepackage[normalem]{ulem}
\usepackage{algpseudocode}
\usepackage{alltt}
\usepackage{fancyvrb}
\usepackage{macros}
\usepackage{wrapfig}
\usepackage{enumitem}
\usepackage{cite}
\usepackage[firstpage]{draftwatermark}

\SetWatermarkText{\hspace*{5in}\raisebox{5.in}{\includegraphics[scale=0.8]{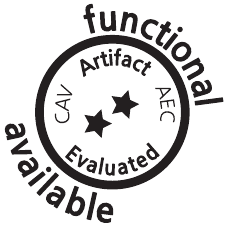}}}
\SetWatermarkAngle{0}

\newcommand{\full}[2]{#2}

\begin{document}

\title{Constraint-based Relational Verification}


\author{Hiroshi Unno\inst{1,2} \and Tachio Terauchi\inst{3} \and Eric Koskinen\inst{4}}
\institute{University of Tsukuba, Ibaraki, Japan \and RIKEN AIP, Tokyo, Japan \and Waseda University, Tokyo, Japan \and Stevens Institute of Technology, New Jersey, USA}
%
%

\maketitle

\vspace{-5mm}

\begin{abstract}
In recent years they have been numerous works that aim to automate relational verification. Meanwhile,  although Constrained Horn Clauses (\CHCS{}) empower a wide range of verification techniques and tools, they lack the ability to express hyperproperties beyond $k$-safety such as generalized non-interference and co-termination.

This paper describes a novel and fully automated constraint-based approach to relational verification. We first introduce 
a new class of predicate Constraint Satisfaction Problems called \PCSPWFFN{} where constraints are represented as clauses modulo first-order theories over predicate variables of three kinds: ordinary, well-founded, or functional. This generalization over \CHCS{}  permits arbitrary (i.e., possibly non-Horn) clauses, well-foundedness constraints,  functionality constraints, and is capable of expressing these relational verification problems.
Our approach enables us to express and automatically verify problem instances that require non-trivial (i.e., non-sequential and non-lock-step) self-composition by automatically inferring appropriate {\em schedulers} (or {\em alignment}) that dictate when and which program copies move.
%
%
To solve problems in this new language,  we present a constraint solving method for \PCSPWFFN{} based on \emph{stratified} CounterExample-Guided Inductive Synthesis (CEGIS) of ordinary, well-founded, and functional predicates.  

We have implemented the proposed framework and obtained promising results on diverse relational verification problems that are beyond the scope of the previous verification frameworks.
\end{abstract}



\section{Introduction}
\label{sec:intro}
We describe a novel constraint-based approach to automatically solving a wide range of relational verification problems including
$k$-safety, co-termination~\cite{DBLP:journals/jlp/Beringer10,Barthe2020}, termination-sensitive non-interference (TS-NI)~\cite{DBLP:conf/csfw/VolpanoS97}, and generalized non-interference (GNI)~\cite{DBLP:conf/sp/McCullough88} for infinite-state programs.

A key challenge in relational property verification is the discovery of {\em relational invariants} which relate the states of multiple program executions.  However, whereas most prior approaches must fix the execution {\em schedule}\footnote{The notion of {\em schedule} is also often called an {\em alignment} in literature.} (e.g., lock-step or sequential)~\cite{Barthe2004,Darvas2005,Terauchi2005,Unno2006,DBLP:conf/esorics/Naumann06,DBLP:journals/toplas/EilersMH20},
a recent work by Shemer et al.~\cite{Shemer2019} has proposed a method to automatically infer
sufficient {\em fair} schedulers to prove the goal relational property.  Importantly, the schedulers in their approach can be {\em semantic} in which the choice of which program to execute can depend on the {\em states} of the programs as opposed to the classic {\em syntactic} schedulers such as lock-step and sequential that can only depend on the control locations.  However, their approach requires the user to provide appropriate atomic predicates and is not fully automatic.  Moreover, they only support $k$-safety properties.
A recent work has proposed a method for automatically  verifying non-hypersafety relational properties but only for {\em finite} state systems~\cite{DBLP:conf/cav/CoenenFST19}.

Meanwhile, today's constraint-based frameworks are also insufficient at automating relational verification. 
The class of predicate constraints called Constrained Horn Clauses (\CHCS{})~\cite{Bjorner2015a} has been widely adopted as a ``common intermediate language'' for uniformly expressing verification problems for various programming paradigms, such as functional and object-oriented languages.  Example uses of the \CHCS{} framework include safety property verification~\cite{Grebenshchikov2012,Gurfinkel2015,Kahsai2016} and refinement type inference~\cite{Unno2009,Terauchi2010,Kobayashi2011b,Jhala2011,Zhu2015}. 
%
The separation of constraint generation and solving has facilitated the rapid development of constraint generation tools such as \rcaml~\cite{Unno2009}, \seahorn~\cite{Gurfinkel2015}, and \jayhorn~\cite{Kahsai2016} as well as efficient constraint solving tools such as \spacer~\cite{Komuravelli2014}, \eldarica~\cite{Hojjat2018}, and \hoice~\cite{Champion2018}.
Unfortunately, \CHCS{} lack the ingredients to sufficiently express these relational verification problems.

In this paper we introduce automated support for relational verification by generalizing \CHCS{} and introducing a new class of predicate Constraint Satisfaction Problems called \PCSPWFFN{}. This language allows constraints that are {\em arbitrary (i.e., possibly non-Horn)} clauses modulo first-order theories over predicate variables that can be {\em functional predicates}, {\em well-founded predicates} or ordinary predicates.  We then show that, thanks to the enhanced predicate variables, \PCSPWFFN{} can express {\em non-hypersafety} relational properties such as co-termination~\cite{DBLP:journals/afp/BeringerH08}, termination-sensitive non-interference (TS-NI)~\cite{DBLP:conf/csfw/VolpanoS97}, and generalized non-interference (GNI)~\cite{DBLP:conf/sp/McCullough88}.  
In addition, our approach effectively quantifies over the schedule, expressing {\em arbitrary fair semantic scheduling} thanks to non-Horn clauses and functional predicates (functional predicates are needed to express fairness in the presence of non-termination which is needed for properties like co-termination and TS-GNI).  The flexibility allows our approach to automatically discover a fair semantic schedule and verify difficult relational problem instances that require non-trivial schedules.  We prove
that our encodings are {\em sound} and {\em complete}. Expressing relational invariants with such flexible scheduling is not possible with \CHCS{}.
However, \PCSPWFFN{} retains a key benefit of \CHCS{}: the idea of separating constraint generation from solving.

We next present a novel constraint solving method for \PCSPWFFN{} based on \emph{stratified} CounterExample-Guided Inductive Synthesis (CEGIS) of ordinary, well-founded, and functional predicates.  In our method, ordinary predicates represent relational inductive invariants, well-founded predicates witness synchronous termination, and functional predicates represent Skolem functions witnessing existential quantifiers that encode angelic non-determinism.
These witnesses for a relational property are often mutually dependent and involve many variables in a complicated way (see \full{Appendix~\ref{app:sol_doublesquareni}, \ref{app:sol_coterm}, and \ref{app:sol_tsgni}}{the extended report~\cite{Unno2021b}} for examples).  The synthesis thus needs to use expressive templates without compromising the efficiency.
Stratified CEGIS combines CEGIS~\cite{Solar-Lezama2006} with stratified families of templates~\cite{Terauchi2015} (\ie, decomposing templates into a series of increasingly expressive templates) to achieve completeness in the sense of~\cite{Jhala2006,Terauchi2015}, a theoretical guarantee of convergence, and a faster and stable convergence by avoiding the overfitting problem of expressive templates to counterexamples~\cite{Padhi2019}.
The constraint solving method naturally generalizes a number of previous techniques developed for \CHCS{} solving and invariant/ranking function synthesis, addressing the challenges due to the generality of \PCSPWFFN{} that is essential for relational verification.

We have implemented the above framework and have applied our tool \pcsat{} to a diverse collection of 20 relational verification problems and obtained promising results.  The benchmark problems go beyond the capabilities of the existing related tools (such as \CHCS{} solvers and program verification tools).  \pcsat{} has solved 15 problems fully automatically by synthesizing complex witnesses for relational properties, and for the 5 problems that could not be solved fully automatically within the time limit, \pcsat{} was able to solve them semi-automatically provided that a part of an invariant is manually given as a hint.


\section{Overview}
\label{sec:overview}
\subsection{Relational verification problems}

\subsubsection{$k$-safety}
\label{sec:overviewksafety}
Consider the following program taken from \cite{Shemer2019} that uses a summation to calculate the square of \texttt{x}, and then doubles it.
\begin{center}
\begin{verbatim}
doubleSquare(bool h, int x) {
   int z, y=0;
   if (h) { z = 2*x; } else { z = x; }
   while (z>0) { z--; y = y+x; }
   if (!h) { y = 2*y; }
   return  y;
}
\end{verbatim}
\end{center}
This program also takes another input  \hh\ and, if the value of \hh\ is true, calculates the result differently. The classical relational property \emph{termination-insensitive non-interference} (TI-NI) says that, roughly, an observer cannot infer the value of high security variables (\hh\ in this case) by observing the outputs (\yy). This is a {\em 2-safety property}~\cite{Terauchi2005,Clarkson2008}: it relates two executions of the same program. In this example, we ask whether two executions that initially agree on \xx\ (\ie, $\xx_1= \xx_2$) will agree on the resulting \yy\ (\ie, $\yy_1 = \yy_2$). The subscripts in these relations indicate copies of the program: $\xx_1$ is variable \xx\ in the first copy of the program and $\xx_2$ is variable \xx\ in the second copy.
More generally, $k$-safety means that if the initial states of a $k$-tuple of programs satisfy a pre-relation $\pre$, then when they all terminate the $k$-tuple of post states will satisfy post-relation $\post$.

The literature proposes many ways to reason about $k$-safety including methods of reducing a multi-program problem to a single-program problem, such as through self-composition~\cite{Barthe2004,Terauchi2005,Unno2006}, product programs~\cite{Barthe2011}, and their variants~\cite{DBLP:journals/toplas/EilersMH20,Sousa2016,Shemer2019,Unno2017b,Pick2018}. Their key challenge is that of \emph{scheduling}: how to interleave the programs' executions so that invariants in the combined program are able to effectively describe cross-program relationships.  Indeed, as proved by \cite{Shemer2019}, verifying this example with the na\"{i}ve lock-step scheduling is impossible with only linear arithmetic invariants while linear arithmetic invariants suffice with a more ``semantic'' scheduling that schedules the copy with $\hh_1 = \false$ to iterate the loop twice per each iteration of the loop in the copy with $\hh_2 = \true$.

In this paper, we will describe a way to pose the scheduling problem as a part of a series of constraints so that the search for an effective scheduler is relegated to the solver level.  In our approach, a $k$-safety verification problem is encoded as a set of constraints containing (ordinary) predicate variables that represent the scheduler to be discovered and a relational invariant preserved by the scheduler.  Specially, we introduce a predicate variable $\textsf{inv}$ that represents a relational invariant and for each $A \subseteq \{1,\ldots,k\}$, a predicate variable $\textsf{sch}_A(\ve{\progvar}_1,\ldots,\ve{\progvar}_k)$ where $\ve{\progvar}_i$ are the variables of the $i$th program, and add constraints that say that if the predicate is $\true$, then the programs whose index are in $A$ will step forward while the rest remain still and also $\textsf{inv}$ is preserved by the step.  For soundness, it is important to constrain the scheduler to be {\em fair}, \ie, at least one program that can progress must be scheduled to progress if there is a program that can progress.  As we shall show in Sec.~\ref{sec:apps}, non-Horn clauses are essential to expressing the fairness constraint.  Roughly, the idea is to use a clause with multiple positive predicate variables (\ie, {\em head disjunction}) to say ``{\em if the relational invariant holds, then at least one of the unfinished programs must be scheduled to progress}.''

Our approach is similar to and is inspired by the approach of~\cite{Shemer2019} that also infers a fair semantic scheduler.  However, their approach requires the user to provide sufficient atomic predicates manually and is not fully automated.  By contrast, our approach soundly-and-completely encodes the $k$-safety verification problem together with scheduler inference as a set of constraints thanks to the expressiveness of \PCSPWFFN{}, and automatically solves those constraints by the stratified CEGIS algorithm (cf.~Sec.~\ref{sec:related} for further comparison).


\subsubsection{Co-termination}
\label{sec:overviewcoterm}

Now consider the following pair of programs.
\[\begin{array}{lll}
 \Pcota &:& \texttt{ while (x>0) \{ x = x - y; \} }\\
 \Pcotb &:&\texttt{ while (x>0) \{ x = x - 2 $\times$ y; \} }
\end{array}\]
A (non-safety) relational question is whether these programs $\Pcota$ and $\Pcotb$ agree on termination~\cite{DBLP:journals/jlp/Beringer10,Barthe2020}. In general they do not: if, for example, $\Pcota$ is executed with $\xx<0$ and $\Pcotb$ with $\xx>0 \wedge \yy=0$, the first will terminate while the second will diverge. However, under the pre-relation $\pre \equiv \xx_1=\xx_2 \wedge \yy_1 = \yy_2$, they will \emph{agree} on termination: the first program terminates iff the second one does.  The property falls outside of the $k$-safety fragment as it cannot be refuted by finite execution traces.  It is worth noting that {\em termination-sensitive non-interference} (TS-NI) is the conjunction of TI-NI and co-termination of two copies of the same target program with $\pre$ equating the copies' low inputs.

Proving co-termination, like $k$-safety, can be aided by scheduler and we can again use our constraints over predicate variables. But this is not enough. We need additional constraints to ensure that whenever one of the two has terminated, the other is also guaranteed to terminate. To address this, we next introduce \emph{well-founded predicate variables}. These predicate variables will appear in our generalized language of constraints  as terms of the form 
$\wfr(\ve{\progvar}_i,\ve{\progvar}_i')$, where the relation $\wfr$ must be \emph{discovered} by the constraint solving method. (In Sec.~\ref{sec:csolve} we describe how to achieve this through our stratified CEGIS algorithm.)   For the above example, our stratified CEGIS algorithm and our tool \pcsat\ automatically discovers (1) a schedule where the two programs step together when $x_1>0$ and $x_2>0$, (2) a relational invariant that implies that if the first program is terminated, then either the second program is terminated or $y_2 \geq 1$ (and vice-versa), and (3) well-founded relations that (combined with the relational invariant) witness that if the loop has terminated in the second program ($x_2 \leq 0$) but not in the first ($x_1 > 0$), then a transition in the first is well-founded (and vice-versa).
%
In Sec.~\ref{sec:apps}, we show how co-termination problems can be soundly-and-completely encoded in \PCSPWFFN{}.

\newcommand{\nd}{*}
\newcommand{\ndi}{*^i}
\newcommand{\ndO}{*^1}

\subsubsection{Generalized non-interference.}
\label{sec:overviewgni}

Now consider the following program.
\begin{center}
\begin{Verbatim}[commandchars=\\\{\},codes={\catcode`$=3\catcode`^=7\catcode`_=8}]
gniEx(bool high, int low) \{
  if (high) \{
    int x = \nd{}$^\text{int}$; if (x >= low) \{ return x; \} else \{ while (true) \{\} \}
  \} else \{
    int x = low; while ($\nd^\text{bool}$) \{ x++; \} return x;
  \}
\}
\end{Verbatim}
\end{center}
The $\nd{}^\texttt{int}$ (resp.~$\nd{}^\texttt{bool}$) above indicates an integer (resp.~a binary) non-deterministic choice.  {\em Termination-insensitive generalized non-interference} (TI-GNI)~\cite{DBLP:conf/sp/McCullough88} is an extension of non-interference to non-deterministic programs, and it says that for any two copies of the program with possibly different values for the high security input (${\tt high}$ in this example) and with the same value for the low security input (${\tt low}$ in this example), if one copy has a terminating execution that ends in some output (the final value of $\xx$ in this example), then the other copy has either a terminating execution ending in the same output or a non-terminating execution.   The {\em termination-sensitive} variant (TS-GNI) strengthens the condition by asserting that if one copy has a terminating execution then the other copy has a terminating execution that ends in the same output.  Both GNI variants are $\forall\exists$ hyperproperties and fall outside of the $k$-safety fragment.

Verifying GNI requires handling non-determinism.  Note that non-determinism occurs both {\em demonically} (\ie, as $\forall$) and {\em angelically} (\ie, as $\exists$) in GNI.  While handling demonic non-determinism is straightforward in a constraint-based verification since the term variables are implicitly universally quantified, handling angelic non-determinism is less straightforward.

Our approach handles finitary angelic non-determinism like $\nd{}^\texttt{bool}$ by adding non-Horn clauses with head disjunctions that roughly express the condition ``{\em the relational invariant remains true in one of the finitely many next step choices}''.  To handle infinitary non-determinism like $\nd{}^\texttt{int}$, we introduce {\em functional predicate variables} denoted $\textsf{f}(\ve{\progvar},r)$. In these terms, \textsf{f} is a predicate variable to be discovered but with a new wrinkle: this predicate involves a return value $r$  and the interpretation of \textsf{f} is a \emph{total function} over $\ve{\progvar}$.  For this example, we introduce the term $\textsf{f}(\ve{\progvar},r)$ where $r$ represents the value chosen non-deterministically at $\nd{}^\texttt{int}$ and $\ve{\progvar}$ are program variables and {\em prophecy variables} that represent the final return values of the demonic copy.  For this example, \pcsat{} automatically discovers the predicate $r = \texttt{ret}_1$ where $\texttt{ret}_1$ is the prophecy variable for the return value of the demonic copy.  With it, \pcsat{} is able to verify TS-GNI and TI-GNI of this example.
We remark that functional predicates are also used to encode scheduler fairness in the presence of non-termination and is needed to ensure soundness for properties like co-termination and TS-GNI.  In Sec.~\ref{sec:appsgni}, we show how TI-GNI and TS-GNI can be soundly-and-completely encoded in \PCSPWFFN{}.

%

\begin{figure*}[t]
\includegraphics[scale=0.45]{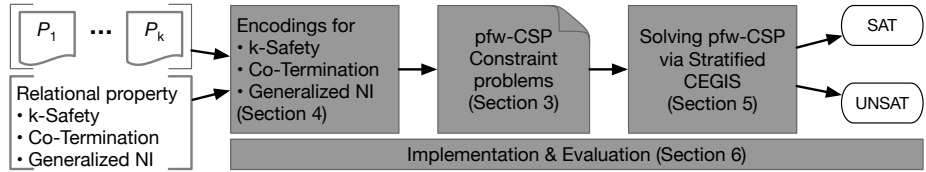}
\caption{\label{fig:diagram} Overview of the contributions and how they achieve a constraint-based strategy for relational verification.}
\end{figure*}

\subsection{Challenges \& Contributions}

There are several challenges that we face in supporting relational verification problems with a constraint-based approach. The subsequent sections of this paper are organized around addressing those challenges as follows: 

\newcommand\mycontrib[2]{\item}

\begin{itemize}
\mycontrib{sec:pcsp}{Generalizing Constraints} 
We first ask how to generalize the constraint language to go beyond CHCs to express a more general class of relational verification problems. To this end,
in Sec.~\ref{sec:pcsp}, we present a new language called \emph{predicated constraint satisfaction problems} (\PCSPWFFN), which incorporate non-Horn clauses, (ordinary) predicate variables, well-founded predicate variables, and functional predicate variables.

\mycontrib{sec:apps}{Encoding Relational Verification} We next return to the above relational verification problems --$k$-safety, co-termination, and generalized non-interference-- and describe how \PCSPWFFN{} can express each of them in a sound and complete manner in Sec.~\ref{sec:apps}.

\mycontrib{sec:csolve}{Solving Constraints} The next major contribution of our research is a novel \emph{stratified} CEGIS algorithm for solving \PCSPWFFN{} constraints.  Our approach integrates advanced verification techniques: {\em stratified family of templates}~\cite{Terauchi2015} and {\em CEGIS of invariants/ranking functions}~\cite{Garg2014,Gonnord2015,Padhi2016,Champion2018}.  While the individual ideas have been proposed previously, they have only been designed for less expressive frameworks such as CHCs, and substantial extensions are needed to combine and apply them to the new \PCSPWFFN{} framework as we shall show in Sec.~\ref{sec:csolve}.

\mycontrib{sec:eval}{Tools} We next turn to an implementation and experimental validation on a diverse collection of 20 relational verification problems, consisting of $k$-safety problems from Shemer \emph{et al.}~\cite{Shemer2019} and new co-termination and GNI problems in Sec.~\ref{sec:eval}.  As far as we know, none of the existing \emph{automated} tools other than our new tool called \pcsat{} can solve them.

\end{itemize}

\noindent
In sum, Fig.~\ref{fig:diagram} depicts each of these sections and how, together, they enable relational verification.  \full{For space, the proofs of the soundness and completeness theorems are deferred to the appendix.}{For space, extra materials are deferred to the extended report~\cite{Unno2021b}.}


\section{Predicate Constraint Satisfaction Problems \PCSPWFFN{}}
\label{sec:pcsp}

As discussed in Sec.~\ref{sec:overview}, \CHCS{} are insufficient to express important relational verification problems.  In the section we introduced a generalized language of constraints called \PCSPWFFN{}. The language of constraint satisfaction problems (CSP) permits non-Horn clauses, {\bf p}redicate variable terms, including those for {\bf f}unctional predicates and {\bf w}ell-founded relations (pfw). 
We now define \PCSPWFFN{}.

Let $\theory$ be a (possibly many-sorted) first-order theory with the signature $\Sigma$.  The syntax of $\theory$-formulas and $\theory$-terms is:
\begin{align*}
  \text{(formulas)} \;
  \phi &::= X(\myseq{t}) \mid p(\myseq{t}) 
  \mid \neg \phi
  \mid \phi_1 \lor \phi_2
  \mid \phi_1 \land \phi_2 \\
  \text{(terms)} \;
  t &::= x \mid f(\myseq{t})
\end{align*}
Here, the meta-variables $x$ and $X$ respectively range over term and predicate variables.
%
%
The meta-variables $p$ and $f$ respectively denote predicate and function symbols of $\Sigma$.  We use $s$ as a meta-variable ranging over sorts of the signature $\Sigma$.  We write $\propos$ for the sort of propositions and $s_1 \to s_2$ for the sort of functions from $s_1$ to $s_2$.  We write $\arity{o}$ and $\sort{o}$ respectively for the arity and the sort of a syntactic element $o$.  A function $f$ represents a constant if $\arity{f}=0$.  We write $\ftv{\phi}$ and $\fpv{\phi}$ respectively for the set of free term and predicate variables that occur in $\phi$.  We write $\myseq{x}$ for a sequence of term variables, $\length{\myseq{x}}$ for the length of $\myseq{x}$, and $\epsilon$ for the empty sequence.  We often abbreviate $\lnot \phi_1 \lor \phi_2$ as $\phi_1 \imply \phi_2$.  We henceforth consider only well-sorted formulas and terms.  We use $\varphi$ as a meta-variable ranging over $\theory$-formulas without predicate variables.

We now define a \PCSP{} $\mathcal{C}$ (with ordinary but without well-founded and functional predicate variables) to be a finite set of clauses of the form
\begin{equation}
\label{eq:pcsp}
\textstyle
\varphi \lor \left( \bigvee_{i=1}^{\ell} X_i(\myseq{t}_i) \right) \lor \left( \bigvee_{i=\ell+1}^m \neg X_i(\myseq{t}_i) \right)
\end{equation}
where $0 \leq \ell \leq m$.
We write $\ftv{c}$ for the set of free term variables of a clause $c$.  The set of free term variables of $\clauses$ is defined by $\ftv{\clauses} = \bigcup_{c \in \clauses} \ftv{c}$.  We regard the variables in $\ftv{c}$ as implicitly universally quantified.  We write $\fpv{\clauses}$ for the set of free predicate variables that occur in $\clauses$.
A \emph{predicate substitution} $\sigma$ is a finite map from predicate variables $X$ to closed predicates of the form $\lambda x_1,\dots,x_{\arity{X}}.\varphi$.  We write $\sigma(\clauses)$ for the application of $\sigma$ to $\clauses$ and $\dom{\sigma}$ for the domain of $\sigma$.  We call $\sigma$ a \emph{syntactic solution} for $\clauses$ if $\fpv{\clauses} \subseteq \dom{\sigma}$ and $\models \bigwedge \sigma(\clauses)$.  Similarly, we call a predicate interpretation $\rho$ a \emph{semantic solution} for $\clauses$ if $\fpv{\clauses} \subseteq \dom{\rho}$ and $\rho \models \bigwedge \clauses$.

\begin{remark}
The language \PCSP{} generalizes over existing languages of constraints.
\CHCS{} can be obtained as a restriction of \PCSP{} where $\ell \leq 1$ in (\ref{eq:pcsp}) for all clauses.  We can also define \COCHCS{} as \PCSP{} but with the restriction that $m \leq \ell+1$ for all clauses.  A linear \CHCS{} is a \PCSP{} that is both \CHCS{} and \COCHCS{}.
\end{remark}

We next extend \PCSP{} to \PCSPWFFN{} by adding well-foundedness and function-ness constraints.  A \PCSPWFFN{} $(\clauses,\kind)$ consists of
\begin{itemize}
    \item a finite set $\clauses$ of \PCSP{}-clauses over predicate variables and
    \item a kinding function $\kind$ that maps each predicate variable $X \in \fpv{\clauses}$ to its kind: any one of $\KORD$, $\KWF$, or $\KFN$ which respectively represent ordinary, well-founded, and functional predicate variables.
\end{itemize}
We write $\rho \models \WF{X}$ if the interpretation $\rho(X)$ of the predicate variable $X$ is \emph{well-founded}, that is, $\sort{X}=(\myseq{s},\myseq{s}) \to \propos$ for some $\myseq{s}$ and there is no infinite sequence $\myseq{v}_1,\myseq{v}_2,\dots$ of sequences $\myseq{v}_i$ of values
of the sorts $\myseq{s}$ such that $(\myseq{v}_i,\myseq{v}_{i+1}) \in \rho(X)$ for all $i\geq 1$.
We write $\rho \models \FN{X}$ if $X$ is {\em functional}, that is, $\sort{X}=(\myseq{s},s) \to \propos$ for some $\myseq{s}$ and $s$, and $\rho \models \allC{\myseq{x}}{\myseq{s}}{(\exiC{y}{s}{X(\myseq{x},y)}) \land \allC{y_1,y_2}{s}{(X(\myseq{x},y_1) \land X(\myseq{x},y_2) \imply y_1=y_2)}}$ holds.
We call a predicate interpretation $\rho$ a \emph{semantic solution} for $(\clauses,\kind)$ if $\rho$ is a semantic solution of $\clauses$, $\rho \models \mathit{WF}(X)$ for all $X$ such that $\kind(X)=\KWF$, and $\rho \models \mathit{FN}(X)$ for all $X$ such that $\kind(X)=\KFN$.
The notion of syntactic solution can be similarly generalized to \PCSPWFFN{}.

\begin{definition}[Satisfiability of \PCSPWFFN{}]
\label{def:psat}
\normalfont
The predicate satisfiability problem of a \PCSPWFFN{} $(\clauses,\kind)$ is 
that of deciding whether it has a semantic solution.
\end{definition}

\begin{remark}
Recall that we assume that the $\theory$-formulas $\varphi$ in \PCSP{} clauses do not contain quantifiers.  The assumption, however, is not a restriction for \PCSPWFFN{} because we can Skolemize quantifiers using functional predicates.

\end{remark}


\section{Relational Verification with Constraints}
\label{sec:apps}

We now present reductions from relational verification problems to \PCSPWFFN{}, thus enabling a new route to automation of these problems. We begin with $k$-safety, and then move toward liveness and non-determinism, which are thorny problems in the relational setting.
We first provide some basic definitions and notations.

\paragraph{Programs.}
We consider programs $P_1$,\dots,$P_k$ on variables
$\ve{\progvar_1}$,\dots,$\ve{\progvar_k}$, respectively.  A {\em state} of the
program $P_i$ is a valuation of the variables $\ve{\progvar_i}$.  We represent
such a valuation by a sequence of values $\ve{v}$ such that $\length{\ve{v}} = \length{\ve{\progvar_i}}$.
We assume that each $P_i$ is defined by the predicate
$T_i(\ve{\progvar_i},\ve{\progvar_i}')$ denoting its one-step transition relation
i.e., $T_i(\ve{v},\ve{v}')$ implies that evaluating $P_i$ one step
from the state $\ve{v}$ reaches the state $\ve{v}'$.  We also assume
that there is a predicate $F_i(\ve{\progvar_i})$ that represents the final
states of the program such that $F_i(\ve{v})$ and
$T_i(\ve{v},\ve{v}')$ implies $\ve{v} = \ve{v}'$, i.e., the program
self-loops when it reaches a final state.  We say that a state
$\ve{v}$ (multi-step) reaches a final state $\ve{v}'$ in the
evaluation of $P_i$, written $\ve{v} \reach_i \ve{v}'$, if there
exists a non-empty finite sequence of states $\pi$ such that $\pi[1]
= \ve{v}$, $\pi[\length{\pi}] = \ve{v}'$, $T_i(\pi[j-1],\pi[j])$ for
all $1 < j \leq \length{\pi}$, and $F_i(\ve{v}')$.  We
write $\ve{v} \reach_i \bot$ if there exists a non-terminating evaluation
from $\ve{v}$ in $P_i$, \ie, if there exists an infinite sequence of
states $\varpi$ such that $\varpi[1] = \ve{v}$, $T_i(\varpi[j-1],\varpi[j])$ for all $1 < j$, and $\neg F_i(\varpi[j])$ for all $0 < j$.  We note that
a program may be non-deterministic, that is, $T_i(\ve{v},\ve{v}')$ and $T_i(\ve{v},\ve{v}'')$ may both be true for some $\ve{v}' \neq \ve{v}''$.

\subsection{$k$-Safety}
\label{sec:appsksafe}


A {\em $k$-safety property} is
given by predicates $\pre(\ve{\progvar})$ and $\post(\ve{\progvar})$ that respectively denote the pre and the post relations across the $k$-tuple.
\begin{definition}[$k$-safety]
\normalfont
  The {\em $k$-safety property verification problem} is to decide if the
following holds:
\[
  \begin{array}{l}
    \forall \ve{v} = \ve{v_1},\dots,\ve{v_k}.\forall \ve{v}' = \ve{v_1}',\dots,\ve{v_k}'.
        \pre(\ve{v}) \wedge \bigwedge_{i\in [k]} \ve{v_i} \reach_i\ve{v_i}' \Rightarrow \post(\ve{v}')
  \end{array}
\]
\end{definition}
\noindent
That is, any $k$-tuple of final states reachable from a $k$-tuple of
states satisfying the precondition satisfies the post condition.  For
instance, the TI-NI verification from Sec.~\ref{sec:overviewksafety} is a
$2$-safety property where $P_1$ and $P_2$ are copies of the same
program, $\pre$ states that the low inputs of the two programs are
equal (\ie, $\xx_1 = \xx_2$ in the example), and $\post$ states that
the low outputs of the two programs are equal (\ie, $\yy_1 = \yy_2$ in
the example).

We now describe a new way to pose the $k$-safety relational verification problem via constraints written in \PCSPWFFN{}. We write $[k]$ for the set $\{1,\dots,k
\}$.  We define $\nepset{[k]} = \{ S \subseteq [k] \mid S \neq \emptyset\}$.  Let $\ve{\progvar} = \ve{\progvar_1}$,\dots,$\ve{\progvar_k}$ be a $k$-tuple of vectors, corresponding to the variables of the $k$ programs. 
\begin{definition}[$k$-safety through constraints]
\label{def:ksafe_encode}
\normalfont
We define \PCSPWFFN{} constraints $\mathcal{C}_\mathrm{S}$ be the set of
following clauses:
\begin{enumerate}
\item[(1)] $\pre(\ve{\progvar}) \Rightarrow \inv(\ve{\progvar})$
\item[(2)] $\inv(\ve{\progvar}) \wedge \bigwedge_{i \in [k]} F_i(\ve{\progvar_i}) \Rightarrow \post(\ve{\progvar})$
\item[(3)] For each $A \in \nepset{[k]}$,\mbox{}\\
  \hspace*{1em}\(
      \inv(\ve{\progvar}) \wedge \sch_A(\ve{\progvar}) \wedge \bigwedge_{i \in A} T_i(\ve{\progvar_i},\ve{\progvar_i}') \wedge \bigwedge_{i \in [k]\setminus A} \ve{\progvar_i} = \ve{\progvar_i}' \Rightarrow \inv(\ve{\progvar}') \)
\item[(4)] For each $A \in \nepset{[k]}$, $\inv(\ve{\progvar}) \wedge \sch_A(\ve{\progvar}) \wedge \bigvee_{i \in [k]} \neg F_i(\ve{\progvar_i}) \Rightarrow \bigvee_{i \in A} \neg F_i(\ve{\progvar_i})$
\item[(5)] $\inv(\ve{\progvar}) \wedge \bigvee_{i \in [k]} \neg F_i(\ve{\progvar_i}) \Rightarrow \bigvee_{A \in \nepset{[k]}} \sch_A(\ve{\progvar})$.
\end{enumerate}
\end{definition}
Here, $\inv$ and $\sch_A$ (for each $A \in \nepset{[k]}$) are ordinary predicate variables.
Roughly, the predicate variables $\sch_A$ describe a {\em scheduler}.
The scheduler stipulates that when $\sch_A(\ve{v_1},\dots,\ve{v_k})$
is true, each $P_i$ such that $i \in A$ takes a step from the state
$\ve{v_i}$ while the others remain still.  Note that the scheduler is
{\em semantic} in the sense that which programs are scheduled to be
executed next can depend on the current states of the programs.
Clauses (1)-(3) assert that $\inv$ is an invariant sufficient to prove
the given safety property with the scheduler defined by $\sch_A$'s.
Clauses (4) say that if an $\inv$-satisfying state is such that the
processes in $A$ are allowed to move and some program has not yet
terminated, then at least one process in $A$ has not yet terminated.
Clause (5) says that any state satisfying $\inv$ has to satisfy some
$\sch_A$.  Clauses (4) and (5) ensure the {\em fairness} of the
scheduler, that is, at least one unfinished program is scheduled to
make progress if there is an unfinished program.

\begin{theorem}[Soundness and Completeness of $\mathcal{C}_\mathrm{S}$]
  \label{thm:ksafety}
  The given $k$-tuple of programs satisfies the given $k$-safety property iff
  $\mathcal{C}_\mathrm{S}$ is satisfiable.
\end{theorem}
We note that the soundness direction crucially relies on scheduler fairness.
The completeness is with respect to semantic solutions (cf. Def.~\ref{def:psat}) and it is only ``relative'' with respect to syntactic solutions: a syntactic solution only exists when the predicates of the background theory are able to express sufficient invariants and schedulers (impossible in general for any decidable theory when the class of programs is Turing-powerful as in our case when the background theory of predicates is QFLIA).

It is important to note that $\mathcal{C}_\mathrm{S}$ is {\em not} \CHCS{} because clause (5) has a head disjunction.
$\mathcal{C}_\mathrm{S}$ may be seen as a constraint-based formulation of the approach proposed in \cite{Shemer2019}.  However, their approach requires the user to provide sufficient predicates manually and is not fully automated, while our approach can fully automatically solve the problems by constraint solving (cf.~Sec.~\ref{sec:csolve}).

\begin{example}
\label{ex:ksafety}
The formalization allows flexible scheduling.  For instance, for the
TI-NI example from Sec.~\ref{sec:overviewksafety}, our approach is able
to infer the predicate substitution that maps $\sch_{\{1\}}$, $\sch_{\{2\}}$, and $\sch_{\{1,2\}}$ to
$\lambda\ve{\progvar}.\hh_1 \wedge \neg \hh_2 \wedge \zz_1+1 = 2\zz_2$,
$\lambda\ve{\progvar}.\neg \hh_1 \wedge \hh_2 \wedge \zz_2+1 = 2\zz_1$, and
$\lambda\ve{\progvar}.(\hh_1 \wedge \neg \hh_2 \Rightarrow \zz_1 = 2\zz_2) \wedge (\neg \hh_1 \wedge \hh_2 \wedge \zz_2 = 2\zz_1)$ respectively, where $\ve{\progvar}$ is the list
of the variables in the two program copies.  The inferred predicates stipulate that
the copy with $\hh = \true$ is scheduled to execute the loop two times per every loop iteration of the copy with $\hh = \false$.
\full{Appendix~\ref{app:doublesquareni}}{The extended report~\cite{Unno2021b}} shows the \PCSPWFFN{} encoding of the example.
A solution generated by \pcsat{} is also shown in \full{Appendix~\ref{app:sol_doublesquareni}.}{\cite{Unno2021b}.}
\end{example}

\subsection{Co-termination}
\label{sec:appscoterm}

Intuitively, co-termination means that if one program terminates, then a second program must terminate~\cite{DBLP:journals/jlp/Beringer10,Barthe2020}. This can also be thought of as a form of relational \emph{termination problem}.\footnote{The property has also been called {\em relative termination}~\cite{DBLP:conf/cade/HawblitzelKLR13}.}
\begin{definition}[Co-Termination]
\normalfont
The {\em co-termination verification problem} is to decide if for all $\ve{v_1},\ve{v_2}$ such that $\pre(\ve{v_1},\ve{v_2})$, if $\ve{v_1} \reach_1 \ve{v_1'}$ then $\ve{v_2} \not\reach_2 \bot$.
\end{definition}
Roughly, the property says that from any pair of states related by $\pre$, if $P_1$ terminates, then $P_2$ must also terminate.
Note that this is an {\em asymmetric} property.  A symmetric version can be obtained by also asserting the property with the positions of the two programs exchanged.  The symmetric version implies, assuming that there is at least one execution from any $\pre$-related state, that from any pair of $\pre$-related states, all executions from one state terminates iff all executions from the other one do as well.
We now present an encoding of conditional co-termination in \PCSPWFFN{}.


\newcommand\Xs{\ve{\progvar}}
\newcommand\twoschTF{\sch_\TF}
\newcommand\twoschTT{\sch_\TT}
\newcommand\twoschFT{\sch_\FT}
\newcommand\twoschp[1]{\textsf{sch}_{#1}}
\begin{definition}[Co-termination through constraints]
\normalfont
\label{def:coterm_encode}
Let $\ve{\progvar} = \ve{\progvar_1}, \ve{\progvar_2}$.
We define \PCSPWFFN{} constraints $\mathcal{C}_\mathrm{CoT}$ be the set of
following clauses:
\begin{enumerate}
\item[(1)] $\pre(\Xs) \wedge \fnbnd(\Xs,b) \Rightarrow \inv(0,b,\Xs)$

\item[(2)] $\inv(d,b,\Xs) \wedge \neg F_1(\ve{\progvar_1}) \wedge \neg F_2(\ve{\progvar_2}) \Rightarrow (-b \leq d \wedge d \leq b \wedge b \geq 0)$
  
\item[(3a)] $\inv(d,b,\Xs) \wedge \twoschFT(d,b,\Xs) \wedge T_2(\ve{\progvar_2},\ve{\progvar_2}') \wedge (F_1(\ve{\progvar_1}) \vee F_2(\ve{\progvar_2}) \vee d' = d-1) \Rightarrow \inv(d',b,\ve{\progvar_1},\ve{\progvar_2}')$

\item[(3b)] $\inv(d,b,\Xs) \wedge \twoschTF(d,b,\Xs) \wedge T_1(\ve{\progvar_1},\ve{\progvar_1}') \wedge (F_1(\ve{\progvar_1}) \vee F_2(\ve{\progvar_2}) \vee d' = d+1) \Rightarrow \inv(d',b,\ve{\progvar_1}',\ve{\progvar_2})$

\item[(3c)]$\inv(d,b,\Xs) \wedge \twoschTT(d,b,\Xs) \wedge T_1(\ve{\progvar_1},\ve{\progvar_1}') \wedge T_2(\ve{\progvar_2},\ve{\progvar_2}') 
\Rightarrow \inv(d,b,\ve{\progvar_1}',\ve{\progvar_2}')$

\item[(4a)] $\inv(d,b,\Xs) \wedge \twoschFT(d,b,\Xs) \wedge \neg F_1(\ve{\progvar_1}) \Rightarrow \neg F_2(\ve{\progvar_2})$

\item[(4b)]$\inv(d,b,\Xs) \wedge \twoschTF(d,b,\Xs) \wedge \neg F_2(\ve{\progvar_2}) \Rightarrow \neg F_1(\ve{\progvar_1})$

\item[(5)]$\inv(d,b,\Xs) \wedge (\neg F_1(\ve{\progvar_1}) \vee \neg F_2(\ve{\progvar_2})) \Rightarrow \bigvee_{a\in \{\TT,\FT,\TF \}} \twoschp{a}(d,b,\Xs)$

\item[(6)] $\inv(d,b,\Xs) \wedge F_1(\ve{\progvar_1}) \wedge \neg F_2(\ve{\progvar_2}) \wedge T_2(\ve{\progvar_2},\ve{\progvar_2}') \Rightarrow \wfr(\ve{\progvar_2},\ve{\progvar_2}')$
\end{enumerate}
\end{definition}
Here, $\sch_\TT$, $\sch_\FT$, and $\sch_\TF$ are 2-specialization of the $k$-safety scheduler of Def.~\ref{def:ksafe_encode}.  Clauses (3x)'s are similar to (3) of Def.~\ref{def:ksafe_encode} and assert that $\inv$ is an invariant under the scheduler.  Clauses (4x)'s and (5), like (4) and (5) of Def.~\ref{def:ksafe_encode}, are used to ensure the scheduler fairness.  However, they are insufficient for co-termination as a non-terminating copy can be scheduled indefinitely leaving the other copy unscheduled.  Clauses (1) and (2) are added to amend the issue.  In (1), $\fnbnd$ is a functional predicate variable that is used to select a {\em bound} $b$, and (2) asserts that the {\em difference} $d$ between the numbers of steps taken by the two copies is within $b$ in any state in $\inv$ when neither copy has terminated.  Note that $d$ is initialized to $0$ by (1) and properly updated in (3x)'s.  Finally, by using the well-founded predicate variable $\wfr$, (6) asserts that if $P_1$ has terminated, then so must eventually $P_2$.

\begin{theorem}[Soundness and Completeness of $\mathcal{C}_\mathrm{CoT}$]
\label{thm:coterm}
The given pair of programs co-terminate iff $\mathcal{C}_\mathrm{CoT}$ is satisfiable.
\end{theorem}
As with Theorem~\ref{thm:ksafety}, the soundness direction relies on scheduler fairness.

\begin{example}
\label{ex:coterm}
Via the encoding, our \pcsat{} tool is able to verify the
symmetric co-termination example from Sec.~\ref{sec:overviewcoterm} by
automatically inferring the solution described there.  For space, the
concrete constraint set and solution are given in \full{Appendix~\ref{app:coterm} and \ref{app:sol_coterm}.}{the extended report~\cite{Unno2021b}.}
\end{example}

\subsection{Generalized Non-Interference}
\label{sec:appsgni}

We now turn to another relational property that cannot simply be captured by $k$-safety or co-termination.  So-called
\emph{termination-insensitive (resp.~-sensitive) generalized non-interference} (resp.~TI-GNI, TS-GNI) are $\forall\exists$ hyperproperties: from any pre-related pair of states whenever one side can take a move to a post state, there must be a way for the other side to also move to a post state such that the post-relation holds. 
As remarked in Sec.~\ref{sec:overview}, verifying GNI requires reasoning about both {\em demonic} (\ie, for all) and {\em angelic} (\ie, exists) {\em non-determinism}.

\begin{definition}[TI/TS-GNI]
\normalfont
The {\em GNI verification problem} is to decide if the following holds.
If $\pre(\ve{v_1},\ve{v_2})$ and $\ve{v_1} \reach_1 \ve{v_1}'$ then
\textbf{(TI-GNI)} $(\exists \ve{v_2}'.\ve{v_2} \reach_2 \ve{v_2}' \wedge \post(\ve{v_1}',\ve{v_2}')) \vee \ve{v_2} \reach_2 \bot$; or \textbf{(TS-GNI)} $\exists \ve{v_2}'.\ve{v_2} \reach_2 \ve{v_2}' \wedge \post(\ve{v_1}',\ve{v_2}')$.
%
\end{definition}
Note that our definition is parameterized by $\pre$ and $\post$.  The standard GNI definitions~\cite{DBLP:conf/sp/McCullough88} can be obtained by letting $P_1$ and $P_2$ be copies of the same target program and letting $\pre$ be the predicate equating the low inputs of the copies and $\post$ be the predicate equating the low outputs of the copies.

To formalize the \PCSPWFFN{} encodings of the GNI verification
problems, we define a relation $U_2$ to be one such that
$T_2(\ve{v},\ve{v}') \Leftrightarrow \exists r. U_2(r,\ve{v},\ve{v}')$
and $U_2(r,\ve{v},\ve{v}') \wedge
U_2(r,\ve{v},\ve{v}'') \Rightarrow \ve{v}' = \ve{v}''$.  
Roughly,
$U_2$ is a function version of the transition relation $T_2$ with the
extra parameter $r$ to make the non-deterministic choices
explicit.

We now show the \PCSPWFFN{} encodings of TI-GNI and TS-GNI.  The key
idea is to augment the encodings for $k$-safety and/or co-termination
with {\em functional predicate variables} and {\em prophecy
variables} that respectively represent the non-deterministic choices of
the angelic side (\ie, $P_2$) and the final outputs of the demonic side (\ie, $P_1$).
\begin{definition}[TI-GNI through constraints]
\label{def:tigni_encode}
\normalfont
We define \PCSPWFFN{} constraints
$\mathcal{C}_\mathrm{TIGNI}$ as $\mathcal{C}_\mathrm{S}$ in
Def.~\ref{def:ksafe_encode} for $k=2$ but with the following modifications:
\begin{itemize}
\item[(m1)] The parameters representing the inputs and outputs of $P_1$ is extended with prophecy variables $\ve{\pfv}$ where $\length{\ve{\pfv}} = \length{\ve{\progvar_1}}$.  Accordingly, each occurrence of $\ve{\progvar_1}$ is replaced by $\ve{\pfv},\ve{\progvar_1}$, and each occurrence of $\ve{\progvar_1}'$ is replaced by $\ve{\pfv}',\ve{\progvar_1}'$.
\item[(m2)]$\pre$ is replaced by $\pre'$ which is defined by
$\pre'(\ve{\pfv},\ve{\progvar_1},\ve{\progvar_2}) \Leftrightarrow \pre(\ve{\progvar_1},\ve{\progvar_2})$, \ie, the prophecy values are unconstrained in the precondition.
\item[(m3)] $F_1$ is replaced by $F_1'$ defined by
$F_1'(\ve{\pfv},\ve{\progvar_1}) \Leftrightarrow F_1(\ve{\progvar_1})$.
\item[(m4)] $T_1$ is replaced by $T_1'$ defined by
$T_1'(\ve{\pfv},\ve{\progvar_1},\ve{\pfv}',\ve{\progvar_1}') \Leftrightarrow T_1(\ve{\progvar_1},\ve{\progvar_1}') \wedge \ve{\pfv} = \ve{\pfv}'$.
\item[(m5)] $\post$ is replaced by $\post'$ defined by
$\post'(\ve{\pfv},\ve{\progvar_1},\ve{\progvar_2}) \Leftrightarrow (\ve{\pfv} = \ve{\progvar_1} \Rightarrow \post(\ve{\progvar_1},\ve{\progvar_2}))$, \ie, if the prophecy was correct then the original post condition must hold.
\item[(m6)] Each occurrence of $T_2(\ve{\progvar_2},\ve{\progvar_2}')$ is replaced by
$\fnr(\ve{\pfv},\ve{\progvar_2},r) \land U_2(r,\ve{\progvar_2},\ve{\progvar_2}')$ where $\fnr$ is a functional predicate variable.
\end{itemize}
\end{definition}
Modifications (m1)-(m5) concern prophecy variables.  They are initialized arbitrarily as shown in (m2), propagated
unmodified through the transitions as shown in (m4), and finally
checked if they match $P_1$'s outputs in (m5).  Modification (m6)
adds functional predicate variables to express the angelic
non-deterministic choices of $P_2$.  The functional predicate
variables shift the onus of making the right choices to the solver's
task of discovering sufficient assignments to them.  Importantly, the
functional predicate takes the prophecy variables as parameters, thus
allowing dependence on the final outputs of the demonic side.

\begin{definition}[TS-GNI through constraints]
\label{def:tsgni_encode}
\normalfont
We define \PCSPWFFN{} constraints
$\mathcal{C}_\mathrm{TSGNI}$ as $\mathcal{C}_\mathrm{CoT}$ in
Def.~\ref{def:coterm_encode} but with modifications of 
Def.~\ref{def:tigni_encode} except (m3) and (m5), and with the following modifications:
\item[(m3')] $F_1$ is replaced by $F_1'$ defined by
$F_1'(\ve{\pfv},\ve{\progvar_1}) \Leftrightarrow F_1(\ve{\progvar_1}) \wedge \ve{\pfv} = \ve{\progvar_1}$.
\item[(m5')] The clause $\inv(\ve{\pfv},\ve{\progvar_1},\ve{\progvar_2}) \wedge F_1'(\ve{\pfv},\ve{\progvar_1}) \wedge F_2(\ve{\progvar_2}) \Rightarrow \post(\ve{\progvar_1},\ve{\progvar_2})$ is added.
\end{definition}
$\mathcal{C}_\mathrm{TSGNI}$ is similar to $\mathcal{C}_\mathrm{TIGNI}$ except that it contains the difference bound and well-foundedness constraints to handle the ``co-termination'' aspect of TS-GNI, \ie, if $P_1$ terminates and makes an output then $P_2$ must also be able terminate and make a matching output.  One subtle aspect of the encoding is that (m3') modifies the final state predicate for $P_1$ to enforce co-termination only when the prophecy is correct.  However, it is worth noting that TS-GNI is {\em not} a conjunction of TI-GNI and co-termination.  For instance, the GNI example from Sec.~\ref{sec:overviewgni} satisfies TS-GNI but does not satisfy co-termination.

\begin{theorem}[Soundess and Completeness of of TI-GNI]
\label{thm:tigni}
The given pair of programs satisfy TI-GNI iff $\mathcal{C}_\mathrm{TIGNI}$ is satisfiable.
\end{theorem}
\begin{theorem}[Soundess and Completeness of TS-GNI]
\label{thm:tsgni}
The given pair of programs satisfy TS-GNI iff $\mathcal{C}_\mathrm{TSGNI}$ is satisfiable.
\end{theorem}
The soundness directions are proven by ``determinizing'' the angelic
choices by solutions to the functional predicate variables and
reducing the argument to those of $k$-safety and co-termination.  The
completeness directions are proven by ``synthesizing'' sufficient angelic
choice functions from program executions.

\begin{example}
\label{ex:tsgni}
Via the encoding, our \pcsat{} tool is able to verify the TS-GNI example from Sec.~\ref{sec:overviewgni} by automatically inferring not only the functional predicate described there but also relational invariants and well-founded relations given in \full{the appendix.}{the extended report~\cite{Unno2021b}.}  For space, the concrete constraint set is also given in \full{the appendix.}{\cite{Unno2021b}.}
\end{example}

\begin{remark}
\label{rem:head_disj}
The angelic non-determinism encoding can be optimized by using head disjunctions when the non-determinism is finitary (\ie, $\max_{\ve{v}} \length{\{ \ve{v}' \mid T_2(\ve{v},\ve{v}')\}}$ is finite) instead of using functional predicate variables.  For this, we modify clauses (3) and (3x)'s of Def.~\ref{def:tigni_encode} and \ref{def:tsgni_encode} to contain multiple positive occurrences of $\inv$ where each occurrence represents one of the finitely many possible choices.
\end{remark}

\begin{remark}
Recall that we allow any program to be non-deterministic.  The $k$-safety and co-termination encodings treat non-determinism in all programs as demonic, whereas the GNI encodings treat those in one program (\ie, $P_1$) as demonic and those in the other program (\ie, $P_2$) as angelic.  In general, an arbitrary program can be made angelic by applying the transformation done in the angelic side of GNI encodings (to factor out non-determinism).
\end{remark}


\section{Constraint Solving Method for \PCSPWFFN{}}
\label{sec:csolve}

We describe a CEGIS-based method for finding a (syntactic) solution of the given \PCSPWFFN{} $(\clauses,\kind)$.
%
%
Our method iterates the following phases until convergence.  The iteration maintains and builds a sequence $\sigma$ of {\em candidate solutions} and a sequence $\examples$ of {\em example instances} where $\examples^{(i)}$ are ground clauses obtained from $\clauses$ by instantiating the term variables and serve as a counterexample to the candidate solution $\sigma^{(i-1)}$, for each $i$-th iteration.
The iteration starts from $\examples^{(1)}=\emptyset$.

{\bf Synthesis Phase}: We check if $(\examples^{(i)},\kind)$ is unsatisfiable.  If so, we stop by returning $\examples^{(i)}$ as a genuine counterexample to the input problem $(\clauses,\kind)$.  Otherwise, we use the synthesizer $\synthTB$ (cf.~Sec.~\ref{sec:synth_tb}) to find a solution $\sigma^{(i)}$ of $(\examples^{(i)},\kind)$, which will be used as the next candidate solution.

{\bf Validation Phase}: We check if $\sigma^{(i)}$ is a genuine solution to $(\clauses,\kind)$ by using an SMT solver.  If so, we stop by returning $\sigma^{(i)}$ as a solution.  Otherwise, for each clause $c \in \clauses$ not satisfied by $\sigma^{(i)}$, we obtain a term substitution $\theta_c$ such that $\dom{\theta_c}=\ftv{c}$ and $\not\models \theta_c(\sigma^{(i)}(c))$.  We then update the example set by adding a new example instance for each unsatisfied clause (i.e., $\examples^{(i+1)}=\examples^{(i)} \cup \myset{\theta_c(c)}{c \in \clauses \land \not\models \sigma^{(i)}(c)}$), and proceed to the next iteration.

The above procedure satisfies the usual {\em progress property} of CEGIS: discovered counterexamples and candidate solutions are not discovered again in succeeding iterations.  Furthermore, as discussed in Sec.~\ref{sec:synth_tb}, by carefully designing the synthesizer $\synthTB$ by incorporating {\em stratified} CEGIS, we achieve {\em completeness} in the sense of \cite{Jhala2006,Terauchi2015}:
 if the given \PCSPWFFN{} $(\clauses,\kind)$ has a syntactic solution expressible in the stratified families of templates, a solution of the \PCSPWFFN{} is eventually found by the procedure.
In Sec.~\ref{sec:synth_tb}, we discuss the details of the synthesis phase.  There, for space, we focus on the theory of quantifier-free linear integer arithmetic (QFLIA).  For space, we defer the details of the unsatisfiability checking process to \full{Appendix~\ref{sec:unsat_ex}.}{the extended report~\cite{Unno2021b}.}

\begin{remark}
The implementation described in Sec.~\ref{sec:eval} contains an additional phase called \emph{resolution phase} for accelerating the convergence.  There, we first apply unit propagation repeatedly to the given $\examples^{(i)}$ to obtain positive examples $\examples^{(i)+}$ of the form $X(\myseq{v})$ and negative examples $\examples^{(i)-}$ of the form $\neg X(\myseq{v})$.  We then repeatedly apply resolution principle to the clauses in the input clauses $\clauses$ and the clauses $\examples^{(i)+} \cup \examples^{(i)-}$ to obtain additional positive and negative examples.
\end{remark}



\subsection{Predicate Synthesis with Stratified Families of Templates}
\label{sec:synth_tb}
\begin{figure}[t]
{\it Stratified Template Family for Ordinary Predicate Variables:}
\vspace{-8pt}
\begin{align*}
\begin{array}{rcl}
\template_X^\bullet(\mathit{nd},\mathit{nc},\mathit{ac},\mathit{ad})
&\defeq&
\lambda (x_1,\dots,x_{\arity{X}}). \bigvee_{i=1}^{\mathit{nd}} \bigwedge_{j=1}^{\mathit{nc}} c_{i,j,0}+\sum_{k=1}^{\arity{X}} c_{i,j,k} \cdot x_k \geq 0 \\
\phi_X^\bullet(\mathit{nd},\mathit{nc},\mathit{ac},\mathit{ad}) &\defeq&
\bigwedge_{i=1}^{\mathit{nd}} \bigwedge_{j=1}^{\mathit{nc}}
( \sum_{k=1}^{\arity{X}} \abs{c_{i,j,k}} \leq \mathit{ac} )
\land
\abs{c_{i,j,0}} \leq \mathit{ad}
\end{array}
\end{align*}

{\it Stratified Template Family for Well-Founded Predicate Variables:}
\vspace{-8pt}
\begin{align*}
\begin{array}{rcl}
\template_X^\Downarrow(\mathit{np},\mathit{nl},\mathit{nc},\mathit{rc},\mathit{rd},\mathit{dc},\mathit{dd})
&\defeq&
\lambda (\myseq{x},\myseq{y}). 
\bigwedge_{i=1}^{\mathit{np}} \bigwedge_{k=1}^{\mathit{nl}} r_{i,k}(\myseq{x}) \geq 0 \land 
(\bigvee_{i=1}^{\mathit{np}} D_i(\myseq{x})) \land \\
\multicolumn{3}{l}{\hspace{10em}
(\bigvee_{j=1}^{\mathit{np}} D_j(\myseq{y})) \land
(\bigvee_{i=1}^{\mathit{np}}
D_i(\myseq{x}) \land
\bigwedge_{j=1}^{\mathit{np}}
(D_j(\myseq{y}) \imply
\mathit{DEC}_{i,j}(\myseq{x},\myseq{y})
)
)
}
 \\
\phi_X^\Downarrow(\mathit{np},\mathit{nl},\mathit{nc},\mathit{rc},\mathit{rd},\mathit{dc},\mathit{dd})
&\defeq&
\bigwedge_{i=1}^{\mathit{np}}
\bigwedge_{k=1}^{\mathit{nl}}
( \sum_{\ell=1}^{\arity{X}/2} \abs{c_{i,k,\ell}} \leq \mathit{rc} ) \land
\abs{c_{i,k,0}} \leq \mathit{rd}\ \land \\
&&\bigwedge_{i=1}^{\mathit{np}}
\bigwedge_{k=1}^{\mathit{nc}}
( \sum_{\ell=1}^{\arity{X}/2} \abs{c_{i,k,\ell}'} \leq \mathit{dc} ) \land
\abs{c_{i,k,0}'} \leq \mathit{dd} \\
\mathit{DEC}_{i,j}(\myseq{x},\myseq{y})
&\defeq&
\bigvee_{k=1}^{\mathit{nl}}
(
r_{i,k}(\myseq{x}) > r_{j,k}(\myseq{y}) \land
\bigwedge_{\ell=1}^{k-1} r_{i,\ell}(\myseq{x}) \geq r_{j,\ell}(\myseq{y})
)
\\
\multicolumn{3}{c}{
r_{i,k}(\myseq{x}) \defeq
c_{i,k,0} + \sum_{\ell=1}^{\arity{X}/2} c_{i,k,\ell} \cdot x_\ell \hspace{1.5em}
D_i(\myseq{x}) \defeq
\bigwedge_{k=1}^{\mathit{nc}}
c_{i,k,0}' + \sum_{\ell=1}^{\arity{X}/2} c_{i,k,\ell}' \cdot x_\ell \geq 0
}
\end{array}
\end{align*}

{\it Stratified Template Family for Functional Predicate Variables:}
\vspace{-8pt}
\begin{align*}
\begin{array}{rcl}
\template_X^\lambda(\mathit{nd},\mathit{nc},\mathit{dc},\mathit{dd},\mathit{ec},\mathit{ed}) &\defeq&
\lambda (\myseq{x},r).
r=\IF D_1(\myseq{x}) \THEN e_1(\myseq{x})
\ELSE \IF D_2(\myseq{x}) \THEN e_2(\myseq{x}) \cdots\\
&&
\ELSE \IF D_{\mathit{nd}-1}(\myseq{x}) \THEN e_{\mathit{nd}-1}(\myseq{x})
\ELSE e_{\mathit{nd}}(\myseq{x})
\\
\phi_X^\lambda(\mathit{nd},\mathit{nc},\mathit{ec},\mathit{ed},\mathit{dc},\mathit{dd}) &\defeq&
\bigwedge_{i=1}^{\mathit{nd}}
( \sum_{j=1}^{\arity{X}-1} \abs{c_{i,j}} \leq \mathit{ec} )
\land
\abs{c_{i,0}} \leq \mathit{ed}
\ \land \\
&&\bigwedge_{i=1}^{\mathit{nd}-1}
\bigwedge_{j=1}^{\mathit{nc}}
( \sum_{k=1}^{\arity{X}-1} \abs{c_{i,j,k}'} \leq \mathit{dc} )
\land
\abs{c_{i,j,0}'} \leq \mathit{dd} \\
\multicolumn{3}{c}{
e_{i}(\myseq{x}) \defeq c_{i,0} + \sum_{j=1}^{\arity{X}-1} c_{i,j} \cdot x_j \hspace{1.5em}
D_{i}(\myseq{x}) \defeq \bigwedge_{j=1}^{\mathit{nc}} c_{i,j,0}' + \sum_{k=1}^{\arity{X}-1} c_{i,j,k}' \cdot x_k \geq 0
}
\end{array}
\end{align*}

\caption{Stratified Families of Templates}
\label{fig:templates}
\end{figure}

We describe our candidate solution synthesizer $\synthTB$.  $\synthTB$
performs a template-based search for a solution to the given example
instances.  As we shall show, our approach allows searching for
assignments to all predicate variables (of all three kinds) in the
given instance which is important because satisfying assignments to
different predicate variables often inter-dependent.
There, however, is a trade-off between expressiveness and
generalizability.  With less expressive templates like intervals, we
may miss actual solutions.  But with very expressive
templates like polyhedra, there could be many solutions, and a
solution thus returned is liable to overfitting, \ie, the
solution to the example instances becomes too specific to be an actual
solution to the original input clauses.  \cite{Padhi2019} discusses a
similar overfitting issue in the context of grammar-based synthesis.

Our remedy to the issue is {\em stratified families of predicate
templates}, inspired by a similar approach proposed in the context of
predicate abstraction with CEGAR~\cite{Jhala2006,Terauchi2015}.
Initially, we assign each predicate variable a less expressive
template and gradually refine it in a counterexample-guided manner: if
no solution exists in the current template, we generate and analyze an
unsat core to identify the \emph{parameters of the families of
templates} that should be updated.
The stratification of templates thus automatically pushes the template to an expressive one (e.g., polyhedra) when it needs to.  Importantly, with our approach, expressive templates are not always used but only when they should be used.

\subsubsection{Stratified Families of Templates}
We have designed three stratified families of templates shown in Fig.~\ref{fig:templates}, respectively for ordinary ($\KORD$), well-founded ($\KWF$), and functional ($\KFN$) predicate variables.
First, for each ordinary predicate variable $X$, we prepare the stratified family of templates $\template_X^\KORD(\mathit{nd},\mathit{nc},\mathit{ac},\mathit{ad})$ with unknowns $c_{i,j,k}$'s to be inferred and its accompanying constraint $\phi_X^\KORD(\mathit{nd},\mathit{nc},\mathit{ac},\mathit{ad})$.  The body of $\template_X^\KORD$ is a DNF with affine inequalities as atoms.
The parameter $\mathit{nd}$ (resp. $\mathit{nc}$) is the number of disjuncts (resp. conjuncts).  The parameter $\mathit{ac}$ is the upper bound of the sum of the absolute values of coefficients $c_{i,j,k}\ (k > 0)$, and $\mathit{ad}$ is the upper bound of the absolute value of $c_{i,j,0}$.

Secondly, for each functional predicate variable $X$, we prepare the stratified family of templates $\template_X^\KWF(\mathit{np},\mathit{nl},\mathit{nc},\mathit{rc},\mathit{rd},\mathit{dc},\mathit{dd})$ with unknowns $c_{i,j,k}$'s and $c_{i,j,k}'$'s and its accompanying constraint $\phi_X^\KWF(\mathit{np},\mathit{nl},\mathit{nc},\mathit{rc},\mathit{rd},\mathit{dc},\mathit{dd})$.  $\template_X^\KWF$ represents the well-founded relation induced by a {\em piecewise-defined lexicographic affine ranking function}~\cite{Alias2010,Urban2014,Leike2015,Urban2013,Leike2015} where $r_{i,j}$ is the affine ranking function template for the $j$-th lexicographic component of the $i$-th region specified by the discriminator $D_i$.  The parameter $\mathit{np}$ (resp. $\mathit{nl}$) is the number of regions (resp. lexicographic components).  The parameters $\mathit{rc},\mathit{rd},\mathit{dc},\mathit{dd}$ are the upper bounds of (the sums of) the absolute values of unknowns,
similar to $\mathit{ac}$ and $\mathit{ad}$ of $\template_X^\KORD$.
%
The first conjunct of $\template_X^\KWF$ asserts that the return value of each ranking functions is non-negative.  The second and the third conjuncts assert that the discriminators cover all relevant states.  Note that discriminators may overlap, and for such overlapping regions, the maximum return value of the ranking functions is used.  The fourth conjunct asserts that the return value of the piecewise-defined ranking function strictly decreases from $\myseq{x}$ to $\myseq{y}$.  Here, $\mathit{DEC}_{i,j}(\myseq{x},\myseq{y})$ asserts that the return value of the lexicographic ranking function for the $i$-th region at $\myseq{x}$ is greater than that for the $j$-th region at $\myseq{y}$.
It follows that for any substitution $\theta$ for the unknowns in $\template_X^\Downarrow$, $\theta(\template_X^\Downarrow)$ represents a well-founded relation.
Our implementation \pcsat{} uses a refined version of $\template_X^\KWF$ shown in \full{Appendix~\ref{sec:wftemp}.}{the extended report~\cite{Unno2021b}.}

Finally, for each functional predicate variable $X$, we prepare the stratified family of templates $\template_X^\KFN(\mathit{nd},\mathit{nc},\mathit{dc},\mathit{dd},\mathit{ec},\mathit{ed})$ with unknowns $c_{i,j}$'s and $c_{i,j,k}'$'s and its accompanying constraint $\phi_F^\KFN(\mathit{nd},\mathit{nc},\mathit{dc},\mathit{dd},\mathit{ec},\mathit{ed})$.  $\template_X^\KFN$ characterizes a piecewise-defined affine function with discriminators $D_1,\dots,D_{\mathit{nd}-1}$ and branch expressions $e_1,\dots,e_{\mathit{nd}}$. The parameter $\mathit{nc}$ is the number of conjuncts in each discriminator.  The parameters $\mathit{dc},\mathit{dd},\mathit{ec},\mathit{ed}$ are the upper bounds of (the sums of) the absolute values of unknown, similar to $\mathit{ac}$ and $\mathit{ad}$ of $\template_X^\KORD$.
Note that for any substitution $\theta$ for the unknowns in $\template_X^\KFN$, $\theta(\template_X^\KFN)(\myseq{x},r)$ expresses a total function that maps $\myseq{x}$ to $r$.

Next, we give the details of the candidate solution synthesis process.  Let $\myseq{p} \in \mathbb{Z}^n$ where $n$ is the number of parameters summed across all templates, and let $\template_X^\alpha(\myseq{p})$ and $\phi_X^\alpha(\myseq{p})$ (for $\alpha \in \finset{\KORD, \KWF, \KFN}$) project the corresponding parameters.  Each $\myseq{p} \in \mathbb{Z}^n$ induces a {\em solution space} $\denote{}{\myseq{p}} \defeq \{ \template(\myseq{p})[\theta] \mid \theta \models \mathit{Con}(\myseq{p})\}$ where $\template(\myseq{p})[\theta] \defeq \{ X \mapsto \theta(\template_X^{\kind(X)}(\myseq{p})) \mid X \in \fpv{\mathcal{C}}\}$ and $\mathit{Con}(\myseq{p}) \defeq \bigwedge_{X \in \fpv{\mathcal{C}}} \phi_X^{\kind(X)}(\myseq{p})$.

Let $\myseq{p_1} \leq \myseq{p_2}$ be the point-wise ordering.  Note that $\denote{}{\myseq{p}}$ is a finite set for any $\myseq{p} \in \mathbb{Z}^n$, and $\myseq{p_1} \leq \myseq{p_2}$ implies $\denote{}{\myseq{p}_1} \subseteq \denote{}{\myseq{p}_2}$.
We start the CEGIS process with some small initial parameters $\myseq{p}^{(0)}$ (the parameters will be maintained as a state of the CEGIS process).  The synthesis phase of each iteration tries to find a solution $\theta \in \denote{}{\myseq{p}^{(i)}}$ to the given example instances $(\examples,\kind)$ where $\myseq{p}^{(i)}$ are the current parameters.  This is done by using an SMT solver for QFLIA to find $\theta$ satisfying $\bigwedge \template(\myseq{p}^{(i)})[\theta](\examples) \wedge \theta(\mathit{Con}(\myseq{p}^{(i)}))$.  If such $\theta$ is found, we return $\template(\myseq{p}^{(i)})[\theta]$ as the candidate solution for the next validation phase of the CEGIS process.  Note that, by construction of the templates, the solution is guaranteed to assign each well-founded (resp.~functional) predicate variable a well-founded relation (resp.~total function).  Otherwise, no solutions to the given example instances $(\examples,\kind)$ can be found in $\denote{}{\myseq{p}^{(i)}}$, and we update the parameters to some $\myseq{p}^{(i+1)} > \myseq{p}^{(i)}$ such that $\denote{}{\myseq{p}^{(i+1)}}$ contains a solution for $(\examples,\kind)$.  Here, it is important to do the update in a {\em fair manner}~\cite{Jhala2006,Terauchi2015}, that is, in any infinite series of updates $\myseq{p}^{(0)},\myseq{p}^{(1)},\dots$, every parameter is updated infinitely often (the details are deferred to below).  By the progress property and the fact that every $\denote{}{\myseq{p}}$ is finite, this ensures that every parameter is updated infinitely often in an infinite series of CEGIS iterations.  We thus obtain the following property.
\begin{theorem}
\label{thm:relcomp}
Our CEGIS-procedure based on stratified families of templates is complete in the sense of \cite{Jhala2006,Terauchi2015}: if there is $\myseq{p}$ and $\sigma \in \denote{}{\myseq{p}}$ such that $\sigma$ is a syntactic solution to the given \PCSPWFFN{} $(\clauses,\kind)$, a syntactic solution to $(\clauses,\kind)$ is eventually found by the procedure.
\end{theorem}
Note that, while the solution space of each stratum (\ie,
$\denote{}{\myseq{p}^{(i)}}$) is finite, our procedure searches the
infinite solution space obtained by taking the infinite union of the
solution spaces of the template family strata (\ie, $\bigcup_{i \in
  \omega}\denote{}{\myseq{p}^{(i)}}$).

\begin{remark}
Our template-based synthesis simultaneously finds ordinary, well-founded, and functional predicates that are mutually dependent through the given $(\examples,\kind)$.  This means that templates for different kinds of predicate variables are updated in a synchronized and balanced manner, which benefits the synthesis of mutually dependent witnesses for a relational property (see \full{Appendix~\ref{app:sol_doublesquareni}, \ref{app:sol_coterm}, and \ref{app:sol_tsgni}}{the extended report~\cite{Unno2021b}} for examples).
\end{remark}

\paragraph{Updating Parameters of Template Families via Unsat Cores.}
We now describe the parameter update process.
We first obtain the unsat core of the unsatisfiability of 
$\bigwedge \template(\myseq{p}^{(i)})[\theta](\examples) \wedge \theta(\mathit{Con}(\myseq{p}^{(i)}))$ from the SMT solver.  We then analyze the core to obtain the parameters of template families, such as the number of conjuncts and disjuncts, that have caused the unsatisfiability.  Here, there could be a dependency between predicate variables and in such a case our unsat core analysis enumerates all the involved predicate variables from which we obtain the parameters of template families to be updated.
We then increment these parameters in some fair manner, by limiting the maximum differences between different parameters to some bounded threshold, and repeatedly solve the resulting constraint until a solution is found.
Thus, the parameters of stratified families of templates are grown on-the-fly guided by the reasons for unsatisfiability.
We found that a careful design of parameter update strategies important for scaling the stratified CEGIS to hard relational verification problems.  The manual tuning, however, is tiresome and suboptimal.  We leave as future work to investigate methods for automatic tuning of parameter update strategies.



\section{Evaluation}
\label{sec:eval}
To evaluate the presented verification framework, we have implemented \pcsat{}, a satisfiability checking tool for \PCSPWFFN{} based on stratified CEGIS.  \pcsat{} supports the theory of Booleans and the quantifier-free theory of linear inequalities over integers and rationals.  The tool is implemented in OCaml, using Z3~\cite{Moura2008} as the backend SMT solver.
%
We ran the tool on a diverse collection of 20 relational verification problems, consisting of $k$-safety,
co-termination, and GNI problems.  Though we have manually reduced them to \PCSPWFFN{} using the presented method in Sec.~\ref{sec:apps}, this process can be easily automated.
The full benchmark set is provided in \full{Appendix~\ref{sec:benchmarks}.}{the extended report~\cite{Unno2021b}.}  All experiments have been conducted on 3.1GHz Intel Xeon Platinum 8000 CPU and 32 GiB RAM with the time limit of 600 seconds.

The experimental results are summarized in Table~\ref{tab:exps}.  The columns ``Time (s)'' and ``\#Iters'' respectively show elapsed wall clock time in seconds and numbers of CEGIS iterations.  \pcsat{} solved 15 verification problems fully automatically and 5 problems labeled with the symbol $\dag$ and/or $\ddag$ semi-automatically.  For the 4 problems labeled with $\dag$, we manually provided small hints for invariant synthesis (interested readers are referred to \full{Appendix~\ref{sec:benchmarks}}{\cite{Unno2021b}}).  The provided hints for all but \verb|SquareSum| are non-relational invariants that can be inferred prior to relational verification by using a \CHCS{} solver or an invariant synthesizer.  For the 2 problems labeled with $\ddag$, we manually chose the initial value for the parameters of the template family for ordinary predicate variables to reduce the elapsed time.  This can be automated by running \pcsat{} with different initial values in parallel.

The problems \verb|DoubleSquareNI_h**|, \verb|HalfSquareNI|, \verb|ArrayInsert|, and \verb|SquareSum| are $k$-safety verification problems obtained from \cite{Shemer2019} that require non-lock-step scheduling.\footnote{We omitted \verb|ArrayIntMod| from \cite{Shemer2019} because its verification requires the theory of arrays which the current version of \pcsat{} does not fully support.}  The problems \verb|DoubleSquareNI_h**| are generated from Example~\ref{ex:ksafety} by a case analysis of the valuation for the boolean variables $h_1$ and $h_2$.  \pcsat{} solved all the $k$-safety problems but \verb|SquareSum| \emph{fully automatically}.  The tool \pdsc{} proposed in \cite{Shemer2019} can verify them but requires the user to provide the atomic predicates for expressing relational invariants and schedulers.
The problems \verb|CotermIntro1| and \verb|CotermIntro2| are asymmetric co-termination problems obtained from the symmetric problem in Example~\ref{ex:coterm}
and are verified by \pcsat{} fully automatically.  
The problems \verb|TS_GNI_h**| are generated from Example~\ref{ex:tsgni} by a case analysis and are verified by \pcsat{} with small non-relational hints.
We have also tested \pcsat{} on various TS-GNI (\verb|SimpleTS_GNI1|, \verb|SimpleTS_GNI2|, \verb|InfBranchTS_GNI|) and TI-GNI problems (\verb|TI_GNI_h**|) and obtained promising results.
As far as we know, no existing tools can verify these non-$k$-safety relational problems.

Furthermore, manual inspection of the \pcsat{}'s output logs for the GNI problems that required hints revealed that the functional predicate synthesis appears to be the main bottleneck of the current version.  In fact, we confirmed that the problems can be solved in less than 10 seconds if appropriate functional predicates for angelic non-determinism are manually provided.  As future work, we plan to investigate methods for improved functional predicate synthesis.

\begin{table*}[t]
\caption{Experimental results of \pcsat{} on the relational verification benchmarks}
\label{tab:exps}
\begin{center}
\begin{tabular}{|l|r|r||l|r|r|}
    \hline
    Program               & Time (s) & \#Iters & Program                & Time (s) & \#Iters \\ \hline\hline
    DoubleSquareNI\_hFT   & 17.762   & 42      & HalfSquareNI           & 11.853   & 35      \\ \hline
    DoubleSquareNI\_hTF   & 26.495   & 55      & ArrayInsert$\ddag$     & 118.671  & 73      \\ \hline
    DoubleSquareNI\_hFF   & 2.944    & 9       & SquareSum$\dag$$\ddag$ & 337.596  & 117     \\ \hline
    DoubleSquareNI\_hTT   & 4.055    & 11      & SimpleTS\_GNI1         & 5.397    & 14      \\ \hline
    CotermIntro1          & 19.322   & 80      & SimpleTS\_GNI2         & 8.919    & 26      \\ \hline
    CotermIntro2          & 15.871   & 73      & InfBranchTS\_GNI       & 2.607    & 4       \\ \hline
    TS\_GNI\_hFT$\dag$    & 47.083   & 78      & TI\_GNI\_hFT$\dag$     & 4.389    & 16      \\ \hline
    TS\_GNI\_hTF          & 5.076    & 17      & TI\_GNI\_hTF           & 2.277    & 6       \\ \hline
    TS\_GNI\_hFF          & 7.174    & 24      & TI\_GNI\_hFF           & 2.968    & 6       \\ \hline
    TS\_GNI\_hTT$\dag$    & 23.495   & 53      & TI\_GNI\_hTT           & 4.148    & 22      \\ \hline
\end{tabular}
\end{center}
\end{table*}


\section{Related Work}
\label{sec:related}
\subsection{Relational Verification}
There has been substantial work on verifying relational properties.
They include program logics, type systems, or program analysis frameworks such as abstract interpretation and model checking~\cite{Volpano1996,DBLP:conf/popl/Benton04,DBLP:conf/cav/FinkbeinerRS15,Sousa2016,DBLP:conf/popl/AssafNSTT17,DBLP:journals/jfp/AguirreBGGS19,DBLP:conf/cav/CoenenFST19}, program transformation approaches such as self-composition or product programs~\cite{DBLP:books/daglib/0078442,Darvas2005,Terauchi2005,Unno2006,DBLP:conf/esorics/Naumann06,DBLP:conf/fm/ZaksP08,Barthe2011,Asada2015,DBLP:conf/pldi/ChurchillP0A19,DBLP:journals/toplas/EilersMH20}, and various other approaches~\cite{Unno2017b,DBLP:conf/pldi/AntonopoulosGHK17,Pick2018,DBLP:conf/cav/FarzanV19,DBLP:journals/pacmpl/ClochardMP20}.  We refer to \cite{DBLP:journals/corr/abs-2007-06421} for an excellent survey.
However, most prior automatic approaches address only the $k$-safety fragment~\cite{Terauchi2005,Clarkson2008} and cannot verify non-$k$-safety (actually, not even hypersafety) properties such as co-termination, TS-NI, TI-GNI, and TS-GNI~\cite{DBLP:journals/afp/BeringerH08,Barthe2020,DBLP:conf/sp/McCullough88}.  The only exception that we are aware is the recent work by Coenen et al.~\cite{DBLP:conf/cav/CoenenFST19} that proposes a sound method for automatically verifying $\forall\exists$ hyperproperties such as GNI for finite state systems.  To our knowledge, we are the {\em first} to propose a sound-and-complete approach to automatically verifying these non-hypersafety properties for infinite state programs.\footnote{However, \cite{DBLP:conf/cav/CoenenFST19} can verify (relational) temporal properties, whereas we only support functional properties that are given by pre and post conditions of whole program runs.  We leave as future work to investigate methods for verifying relational temporal properties of infinite state programs.}

A key task in many relational verification methods, including ours, is
the discovery of {\em relational invariants} which relate the states
of multiple program executions.  While most prior approaches are
limited to fixed execution schedule (or {\em alignment}) such as
lock-step and
sequential~\cite{Barthe2004,Darvas2005,Terauchi2005,Unno2006,DBLP:conf/esorics/Naumann06,Barthe2011,DBLP:journals/toplas/EilersMH20},
a recent work by Shemer et al.~\cite{Shemer2019} has proposed a
$k$-safety property verification method that automatically infers
fair schedulers sufficient to prove
the goal property.  Importantly, the schedulers in their approach can be {\em semantic} in which the choice of which program to execute can depend on the {\em states} of the programs as opposed to the classic {\em syntactic} schedulers such as lock-step and sequential that can only depend on the control locations.  Our approach also infers such fair semantic schedulers, and as remarked before, they enable
solving instances like \textsf{doubleSquare} that are difficult for
previous approaches.  However, \cite{Shemer2019} requires the user to
provide appropriate atomic predicates and is not fully automatic.  By
contrast, our approach soundly and completely encodes the problem as a
constraint satisfaction problem and fully automatically verifies hard
instances like
\textsf{doubleSquare} by constraint solving.

Furthermore, our work extends the fair semantic scheduler
synthesis to beyond $k$-safety problems like co-termination, TI-GNI
and TS-GNI, in a sound and complete manner.  We note that the
extensions are non-trivial and involves delicate uses of functional
predicate variables and well-founded predicate variables to ensure
scheduler fairness in the presence of non-termination as well as uses
of prophecy variables and functional predicate variables to handle
angelic non-determinism.  The higher-degree of automation and the
extension to non-$k$-safety properties are thanks to the expressive
power of our novel constraint framework \PCSPWFFN{}.

\subsection{Predicate Constraint Solving}

Our \PCSPWFFN{} solving technique builds on and generalizes a number of techniques developed for \CHCS{} solving as well as invariant and ranking function discovery.  Most closely related to our constraint solving method are CEGIS-based~\cite{Solar-Lezama2006} and data-driven approaches to solving \CHCS{}~\cite{Sharma2013a,Sharma2013,Garg2014,Krishna2015,Garg2016,Padhi2016,Champion2018,Ezudheen2018,Zhu2018,Fedyukovich2018,Padhi2019}.  As remarked before, the new \PCSPWFFN{} framework is strictly more expressive than \CHCS{} and extending the prior techniques to the new framework is non-trivial.

Our stratified CEGIS is inspired by the idea of stratified languages of predicates proposed in the context of predicate abstraction with CEGAR~\cite{Jhala2006,Terauchi2015}.  It is also similar in spirit to the work by Padhi et al.~\cite{Padhi2019}, but they use a stratified family of grammars.  Also none of these prior works use unsat cores for updating the language/grammar stratum, synthesize well-founded relations and functional predicates, or support non-Horn clauses.

%
%

Our class of \PCSPWFFN{} constraints is related to {\em
existentially-quantified Horn clauses} (E-CHCs) introduced by Beyene
et al.~\cite{Beyene2013}.  E-CHCs does not have non-Horn clauses or
functional predicate variables.  However, it has disjunctive well-foundedness
constraints which are similar to our well-founded predicate variables.
Also, existential quantifiers can be used to encode head
disjunctions and functional predicates.
We conjecture that \PCSPWFFN{} and E-CHCs are inter-reducible, but it
is not trivial to fill the gap.  Also, inter-reducibility is a
desirable feature: different formats may have different benefits.  For
relational verification, as we have shown, \PCSPWFFN{} enables direct
sound-and-complete encodings of the problems.  For instance, head
disjunctions allow direct encoding of scheduler fairness and
finitary angelic non-determinism (cf.~Remark~\ref{rem:head_disj}).
And, functional predicate variables can be explicitly given
necessary-and-sufficient parameters to encode angelic non-determinism
and difference bounds for ensuring scheduler fairness in the presence
of non-termination.  The tight encodings also lead to reduction in
search space and benefited the constraint solving.

\section{Conclusion}
\label{sec:conc}
We have introduced the class \PCSPWFFN{} of predicate constraint satisfaction problems that generalizes \CHCS{} with arbitrary clauses, well-foundedness constraints, and functionality constraints.  We have then established a program verification framework based on \PCSPWFFN{} by showing that (1) 
\PCSPWFFN{} can soundly-and-completely encode various classes of relational problems of infinite-state non-deterministic programs, including hard instances of $k$-safety, co-termination, and termination-sensitive generalized non-interference that benefit from state-dependent scheduling/alignment (Theorems~\ref{thm:ksafety}--\ref{thm:tsgni}), and (2) existing \CHCS{} solving and invariants/ranking function synthesis techniques can be adopted to \PCSPWFFN{} solving and further improved with the idea of stratified CEGIS for simultaneously achieving completeness (Theorem~\ref{thm:relcomp}) and practical effectiveness.

In future work we plan to investigate ways 
to improve functional predicate synthesis,
automatic tuning of parameter update strategies for constraint solving,
and whether a constraint-based approach (and the techniques presented in the present paper) can be extended to reason about relational temporal properties such as the ones expressed in hyper temporal logics~\cite{DBLP:conf/post/ClarksonFKMRS14,DBLP:conf/cav/FinkbeinerRS15}.

\emph{Acknowledgments.}
We thank the anonymous reviewers for their suggestions.
This work was supported by ONR grant \# N00014-17-1-2787, JST ERATO HASUO Metamathematics for Systems Design Project (No. JPMJER1603), and JSPS KAKENHI Grant Numbers 17H01720, 18K19787, 19H04084, 20H04162, 20H05703, and 20K20625.
%


\bibliographystyle{splncs04}
\bibliography{abbrv,prog_lang}

\begin{thebibliography}{10}
\providecommand{\url}[1]{\texttt{#1}}
\providecommand{\urlprefix}{URL }
\providecommand{\doi}[1]{https://doi.org/#1}

\bibitem{DBLP:journals/jfp/AguirreBGGS19}
Aguirre, A., Barthe, G., Gaboardi, M., Garg, D., Strub, P.: A relational logic
  for higher-order programs. J. Funct. Program.  \textbf{29} (2019)

\bibitem{Alias2010}
Alias, C., Darte, A., Feautrier, P., Gonnord, L.: Multi-dimensional rankings,
  program termination, and complexity bounds of flowchart programs. In: SAS
  '10. pp. 117--133. Springer (2010)

\bibitem{DBLP:conf/pldi/AntonopoulosGHK17}
Antonopoulos, T., Gazzillo, P., Hicks, M., Koskinen, E., Terauchi, T., Wei, S.:
  Decomposition instead of self-composition for proving the absence of timing
  channels. In: PLDI (2017)

\bibitem{Asada2015}
Asada, K., Sato, R., Kobayashi, N.: Verifying relational properties of
  functional programs by first-order refinement. In: PEPM (2015)

\bibitem{DBLP:conf/popl/AssafNSTT17}
Assaf, M., Naumann, D.A., Signoles, J., Totel, E., Tronel, F.: Hypercollecting
  semantics and its application to static analysis of information flow. In:
  POPL (2017)

\bibitem{Barthe2020}
Barthe, G.: An introduction to relational program verification (2020)

\bibitem{Barthe2011}
Barthe, G., Crespo, J.M., Kunz, C.: Relational verification using product
  programs. In: FM (2011)

\bibitem{Barthe2004}
Barthe, G., D'Argenio, P.R., Rezk, T.: Secure information flow by
  self-composition. In: CSFW (2004)

\bibitem{DBLP:conf/popl/Benton04}
Benton, N.: Simple relational correctness proofs for static analyses and
  program transformations. In: POPL (2004)

\bibitem{DBLP:journals/jlp/Beringer10}
Beringer, L.: Relational bytecode correlations. J. Log. Alg. Meth. Pro.
  \textbf{79}(7) (2010)

\bibitem{DBLP:journals/afp/BeringerH08}
Beringer, L., Hofmann, M.: Secure information flow and program logics. Arch.
  Formal Proofs  (2008)

\bibitem{Beyene2013}
Beyene, T.A., Popeea, C., Rybalchenko, A.: Solving existentially quantified
  {Horn} clauses. In: CAV (2013)

\bibitem{Bjorner2015a}
Bj{\o}rner, N., Gurfinkel, A., McMillan, K.L., Rybalchenko, A.: {Horn} clause
  solvers for program verification. In: Fields of Logic and Computation {II}:
  Essays Dedicated to Yuri Gurevich on the Occasion of His 75th Birthday (2015)

\bibitem{Champion2018}
Champion, A., Chiba, T., Kobayashi, N., Sato, R.: {ICE}-based refinement type
  discovery for higher-order functional programs. In: TACAS (2018)

\bibitem{DBLP:conf/pldi/ChurchillP0A19}
Churchill, B.R., Padon, O., Sharma, R., Aiken, A.: Semantic program alignment
  for equivalence checking. In: PLDI (2019)

\bibitem{DBLP:conf/post/ClarksonFKMRS14}
Clarkson, M.R., Finkbeiner, B., Koleini, M., Micinski, K.K., Rabe, M.N.,
  S{\'{a}}nchez, C.: Temporal logics for hyperproperties. In: POST (2014)

\bibitem{Clarkson2008}
Clarkson, M.R., Schneider, F.B.: Hyperproperties. In: CSF (2008)

\bibitem{DBLP:journals/pacmpl/ClochardMP20}
Clochard, M., March{\'{e}}, C., Paskevich, A.: Deductive verification with
  ghost monitors. PACMPL  \textbf{4}({POPL}) (2020)

\bibitem{DBLP:conf/cav/CoenenFST19}
Coenen, N., Finkbeiner, B., S{\'{a}}nchez, C., Tentrup, L.: Verifying
  hyperliveness. In: CAV (2019)

\bibitem{Darvas2005}
\'{A}d\'{a}m Darvas, H\"{a}hnle, R., Sands, D.: A theorem proving approach to
  analysis of secure information flow. In: SPC (2005)

\bibitem{DBLP:journals/toplas/EilersMH20}
Eilers, M., M{\"{u}}ller, P., Hitz, S.: Modular product programs. TOPLAS
  \textbf{42}(1) (2020)

\bibitem{Ezudheen2018}
Ezudheen, P., Neider, D., D'Souza, D., Garg, P., Madhusudan, P.: {Horn-ICE}
  learning for synthesizing invariants and contracts. PACMPL
  \textbf{2}(OOPSLA) (2018)

\bibitem{DBLP:conf/cav/FarzanV19}
Farzan, A., Vandikas, A.: Automated hypersafety verification. In: CAV (2019)

\bibitem{Fedyukovich2018}
Fedyukovich, G., Zhang, Y., Gupta, A.: Syntax-guided termination analysis. In:
  CAV (2018)

\bibitem{DBLP:conf/cav/FinkbeinerRS15}
Finkbeiner, B., Rabe, M.N., S{\'{a}}nchez, C.: Algorithms for model checking
  {HyperLTL} and {HyperCTL$^*$}. In: CAV (2015)

\bibitem{Garg2014}
Garg, P., L{\"o}ding, C., Madhusudan, P., Neider, D.: {ICE}: A robust framework
  for learning invariants. In: CAV (2014)

\bibitem{Garg2016}
Garg, P., Neider, D., Madhusudan, P., Roth, D.: Learning invariants using
  decision trees and implication counterexamples. In: POPL (2016)

\bibitem{Gonnord2015}
Gonnord, L., Monniaux, D., Radanne, G.: Synthesis of ranking functions using
  extremal counterexamples. In: PLDI (2015)

\bibitem{Grebenshchikov2012}
Grebenshchikov, S., Lopes, N.P., Popeea, C., Rybalchenko, A.: Synthesizing
  software verifiers from proof rules. In: PLDI (2012)

\bibitem{Gurfinkel2015}
Gurfinkel, A., Kahsai, T., Komuravelli, A., Navas, J.A.: The {SeaHorn}
  verification framework. In: CAV (2015)

\bibitem{DBLP:conf/cade/HawblitzelKLR13}
Hawblitzel, C., Kawaguchi, M., Lahiri, S.K., Reb{\^{e}}lo, H.: Towards
  modularly comparing programs using automated theorem provers. In: CADE (2013)

\bibitem{Hojjat2018}
Hojjat, H., R{\"u}mmer, P.: The {Eldarica} horn solver. In: FMCAD (2018)

\bibitem{Jhala2011}
Jhala, R., Majumdar, R., Rybalchenko, A.: {HMC}: verifying functional programs
  using abstract interpreters. In: CAV (2011)

\bibitem{Jhala2006}
Jhala, R., McMillan, K.L.: A practical and complete approach to predicate
  refinement. In: TACAS (2006)

\bibitem{Kahsai2016}
Kahsai, T., R{\"u}mmer, P., Sanchez, H., Sch{\"a}f, M.: {JayHorn}: A framework
  for verifying {Java} programs. In: CAV (2016)

\bibitem{Kobayashi2011b}
Kobayashi, N., Sato, R., Unno, H.: Predicate abstraction and {CEGAR} for
  higher-order model checking. In: PLDI (2011)

\bibitem{Komuravelli2014}
Komuravelli, A., Gurfinkel, A., Chaki, S.: {SMT}-based model checking for
  recursive programs. In: CAV (2014)

\bibitem{Krishna2015}
Krishna, S., Puhrsch, C., Wies, T.: Learning invariants using decision trees.
  CoRR  \textbf{abs/1501.04725} (2015)

\bibitem{Leike2015}
Leike, J., Heizmann, M.: Ranking templates for linear loops. LMCS
  \textbf{11}(1) (2015)

\bibitem{DBLP:conf/sp/McCullough88}
McCullough, D.: Noninterference and the composability of security properties.
  In: SP (1988)

\bibitem{Moura2008}
de~Moura, L., Bj{\o}rner, N.: Z3: An efficient {SMT} solver. In: TACAS (2008)

\bibitem{DBLP:conf/esorics/Naumann06}
Naumann, D.A.: From coupling relations to mated invariants for checking
  information flow. In: ESORICS (2006)

\bibitem{DBLP:journals/corr/abs-2007-06421}
Naumann, D.A.: Thirty-seven years of relational hoare logic: remarks on its
  principles and history. CoRR  \textbf{abs/2007.06421} (2020)

\bibitem{Padhi2019}
Padhi, S., Millstein, T.D., Nori, A.V., Sharma, R.: Overfitting in synthesis:
  Theory and practice. In: CAV (2019)

\bibitem{Padhi2016}
Padhi, S., Sharma, R., Millstein, T.D.: Data-driven precondition inference with
  learned features. In: PLDI (2016)

\bibitem{Pick2018}
Pick, L., Fedyukovich, G., Gupta, A.: Exploiting synchrony and symmetry in
  relational verification. In: CAV (2018)

\bibitem{DBLP:books/daglib/0078442}
Reynolds, J.C.: The craft of programming. Prentice Hall (1981)

\bibitem{Sharma2013a}
Sharma, R., Gupta, S., Hariharan, B., Aiken, A., Liang, P., Nori, A.V.: A data
  driven approach for algebraic loop invariants. In: ESOP (2013)

\bibitem{Sharma2013}
Sharma, R., Gupta, S., Hariharan, B., Aiken, A., Nori, A.V.: Verification as
  learning geometric concepts. In: SAS (2013)

\bibitem{Shemer2019}
Shemer, R., Gurfinkel, A., Shoham, S., Vizel, Y.: Property directed self
  composition. In: CAV (2019)

\bibitem{Solar-Lezama2006}
Solar-Lezama, A., Tancau, L., Bodik, R., Seshia, S., Saraswat, V.:
  Combinatorial sketching for finite programs. In: ASPLOS (2006)

\bibitem{Sousa2016}
Sousa, M., Dillig, I.: Cartesian hoare logic for verifying k-safety properties.
  In: PLDI (2016)

\bibitem{Terauchi2010}
Terauchi, T.: Dependent types from counterexamples. In: POPL (2010)

\bibitem{Terauchi2005}
Terauchi, T., Aiken, A.: Secure information flow as a safety problem. In: SAS
  (2005)

\bibitem{Terauchi2015}
Terauchi, T., Unno, H.: Relaxed stratification: A new approach to practical
  complete predicate refinement. In: ESOP (2015)

\bibitem{Unno2009}
Unno, H., Kobayashi, N.: Dependent type inference with interpolants. In: PPDP
  (2009)

\bibitem{Unno2006}
Unno, H., Kobayashi, N., Yonezawa, A.: Combining type-based analysis and model
  checking for finding counterexamples against non-interference. In: PLAS
  (2006)

\bibitem{Unno2021b}
Unno, H., Terauchi, T., Koskinen, E.: Constraint-based relational verification.
  Full version, available from \url{http://www.cs.tsukuba.ac.jp/~uhiro/} (2021)

\bibitem{Unno2017b}
Unno, H., Torii, S., Sakamoto, H.: Automating induction for solving horn
  clauses. In: CAV (2017)

\bibitem{Urban2013}
Urban, C.: The abstract domain of segmented ranking functions. In: SAS (2013)

\bibitem{Urban2014}
Urban, C., Min{\'e}, A.: An abstract domain to infer ordinal-valued ranking
  functions. In: ESOP (2014)

\bibitem{Volpano1996}
Volpano, D.M., Irvine, C., Smith, G.: A sound type system for secure flow
  analysis. J. Comp. Sec.  \textbf{4}(2-3) (1996)

\bibitem{DBLP:conf/csfw/VolpanoS97}
Volpano, D.M., Smith, G.: Eliminating covert flows with minimum typings. In:
  CSFW (1997)

\bibitem{DBLP:conf/fm/ZaksP08}
Zaks, A., Pnueli, A.: {CoVaC}: Compiler validation by program analysis of the
  cross-product. In: FM (2008)

\bibitem{Zhu2018}
Zhu, H., Magill, S., Jagannathan, S.: A data-driven {CHC} solver. In: PLDI
  (2018)

\bibitem{Zhu2015}
Zhu, H., Nori, A.V., Jagannathan, S.: Learning refinement types. In: ICFP
  (2015)

\end{thebibliography}

\full{
\newpage
\appendix    
\section{The encoding of TI-NI verification from Example~\ref{ex:ksafety}}
\label{app:doublesquareni}

\[
\begin{array}{rcl}
  \inv(\ve{\progvar}_1,\ve{\progvar}_2) & \Leftarrow & x_1 = x_2\ \wedge \\
  && y_1 = 0 \wedge (h_1 \wedge z_1 = 2\times x_1 \vee \neg h_1 \wedge z_1 = x_1)\  \wedge \\
  && y_2 = 0 \wedge (h_2 \wedge z_2 = 2\times x_2 \vee \neg h_2 \wedge z_2 = x_2) \\
  \inv(\ve{\progvar}_1',\ve{\progvar}_2) & \Leftarrow & \inv(\ve{\progvar}_1,\ve{\progvar}_2) \wedge \sch_\TF(\ve{\progvar}_1,\ve{\progvar}_2)\ \wedge \\
  & &  (z_1 > 0 \wedge z_1' = z_1-1 \wedge y_1' = y_1+x_1 \vee z_1 \leq 0 \wedge z_1' = z_1 \wedge y_1' = y_1) \\
  \inv(\ve{\progvar}_1,\ve{\progvar}_2') & \Leftarrow & \inv(\ve{\progvar}_1,\ve{\progvar}_2) \wedge \sch_\FT(\ve{\progvar}_1,\ve{\progvar}_2)\ \wedge \\
  & &  (z_2 > 0 \wedge z_2' = z_2-1 \wedge y_2' = y_2+x_2 \vee z_2 \leq 0 \wedge z_2' = z_2 \wedge y_2' = y_2) \\
  \inv(\ve{\progvar}_1',\ve{\progvar}_2') & \Leftarrow & \inv(\ve{\progvar}_1,\ve{\progvar}_2) \wedge \sch_\TT(\ve{\progvar}_1,\ve{\progvar}_2)\ \wedge \\
  & & (z_1 > 0 \wedge z_1' = z_1-1 \wedge y_1' = y_1+x_1 \vee z_1 \leq 0 \wedge z_1' = z_1 \wedge y_1' = y_1)\ \wedge \\
  & & (z_2 > 0 \wedge z_2' = z_2-1 \wedge y_2' = y_2+x_2 \vee z_2 \leq 0 \wedge z_2' = z_2 \wedge y_2' = y_2) \\
  z_1 > 0 & \Leftarrow & \inv(\ve{\progvar}_1,\ve{\progvar}_2) \wedge \sch_\TF(\ve{\progvar}_1,\ve{\progvar}_2) \wedge z_2 > 0 \\
  z_2 > 0 & \Leftarrow & \inv(\ve{\progvar}_1,\ve{\progvar}_2) \wedge \sch_\FT(\ve{\progvar}_1,\ve{\progvar}_2) \wedge z_1 > 0 \\
  \multicolumn{3}{l}{\sch_\TF(\ve{\progvar}_1,\ve{\progvar}_2) \vee \sch_\FT(\ve{\progvar}_1,\ve{\progvar}_2) \vee \sch_\TT(\ve{\progvar}_1,\ve{\progvar}_2) \Leftarrow \inv(\ve{\progvar}_1,\ve{\progvar}_2) \wedge (z_1 > 0 \vee z_2 > 0)} \\
  y_1' = y_2' & \Leftarrow & \inv(\ve{\progvar}_1,\ve{\progvar}_2) \wedge z_1 \leq 0 \wedge z_2 \leq 0\ \wedge \\
  & & (h_1 \wedge y_1' = y_1 \vee \neg h_1 \wedge y_1' = 2\times y_1)\  \wedge \\
  & &  (h_2 \wedge y_2' = y_2 \vee \neg h_2 \wedge y_2' = 2\times y_2)
\end{array}
\]
where $\ve{\progvar}_1 = x_1,y_1,z_1,h_1$, $\ve{\progvar}_1' = x_1,y_1',z_1',h_1'$, $\ve{\progvar}_2 = x_2,y_2,z_2,h_2$, $\ve{\progvar}_2' = x_2,y_2',z_2',h_2'$, $\TF = \{1\}$, $\FT = \{2\}$, and $\TT = \{1,2\}$.

\section{The \pcsat{} generated solution of Example~\ref{ex:ksafety}}
\label{app:sol_doublesquareni}

\[
\begin{array}{rcl}
  \inv(\ve{\progvar}_1,\ve{\progvar}_2) & \equiv &
  \left(\begin{array}{l}
   \neg h_1 \land \neg h_2 \land \left(\begin{array}{l}
    x_1 = x_2 \land y_1 = y_2 \land z_1 = z_2\ \lor \\
    x_2 + y_2 = 0 \land z_1 < 0\ \land \\
    1 + 2 \cdot x_2 = 2 \cdot z_2 \land 2 + y_1 = y_2
   \end{array}\right) \lor \\
   \neg h_1 \land h_2 \land \left(\begin{array}{l}
    x_1 = x_2 \land 1 + 2 \cdot x_1 \neq z_2 \land 2 \cdot y_1 = y_2 \land 2 \cdot z_1 = z_2\ \lor \\
    x_1 = x_2 \land 2 \cdot x_1 \neq z_2 \land z_1 \geq 1 \land z_2 \geq 1\ \land \\
    x_1 + 2 \cdot y_1 = y_2 \land 2 \cdot z_1 = 1 + z_2
   \end{array}\right) \lor \\
   h_1 \land \neg h_2 \land \left(\begin{array}{l}
    x_1 = x_2 \land y_1 = 2 \cdot y_2 \land z_1 = 2 \cdot z_2\ \lor \\
    x_1 = x_2 \land 2 \cdot x_1 \neq z_1 \land z_1 \geq 1 \land z_2 \geq 1\ \land \\
    y_1 = x_2 + 2 \cdot y_2 \land 1 + z_1 = 2 \cdot z_2
  \end{array}\right) \lor \\
   h_1 \land h_2 \land
    x_1 = x_2 \land y_1 = y_2 \land z_2 = z_2
   \end{array}\right)
\end{array}
\]

\section{The encoding of the co-termination verification from Example~\ref{ex:coterm}}
\label{app:coterm}



\[
\begin{array}{rcl}
  \inv(0,b,\ve{\progvar}_1,\ve{\progvar}_2) & \Leftarrow & \fnbnd(\ve{\progvar}_1,\ve{\progvar}_2,b) \wedge x_1 = x_2 \wedge y_1 = y_2 \\
  \inv(d',b,\ve{\progvar}_1',\ve{\progvar}_2) & \Leftarrow & \inv(d,b,\ve{\progvar}_1,\ve{\progvar}_2) \wedge \sch_\TF(d,b,\ve{\progvar}_1,\ve{\progvar}_2)\ \wedge \\
  & &  (x_1 > 0 \wedge x_1' = x_1-y_1 \vee x_1 \leq 0 \wedge x_1' = x_1)\ \wedge \\
  & &  (x_1 \leq 0 \vee x_2 \leq 0 \vee d' = d + 1) \\
  \inv(d',b,\ve{\progvar}_1,\ve{\progvar}_2') & \Leftarrow & \inv(d,b,\ve{\progvar}_1,\ve{\progvar}_2) \wedge \sch_\FT(d,b,\ve{\progvar}_1,\ve{\progvar}_2)\ \wedge \\
  & &  (x_2 > 0 \wedge x_2' = x_2-2 \cdot y_2 \vee x_2 \leq 0 \wedge x_2' = x_2)\ \wedge \\
  & &  (x_1 \leq 0 \vee x_2 \leq 0 \vee d' = d - 1) \\
  \inv(d,b,\ve{\progvar}_1',\ve{\progvar}_2') & \Leftarrow & \inv(d,b,\ve{\progvar}_1,\ve{\progvar}_2) \wedge \sch_\TT(d,b,\ve{\progvar}_1,\ve{\progvar}_2)\ \wedge \\
  & &  (x_1 > 0 \wedge x_1' = x_1-y_1 \vee x_1 \leq 0 \wedge x_1' = x_1)\ \wedge \\
  & &  (x_2 > 0 \wedge x_2' = x_2-2 \cdot y_2 \vee x_2 \leq 0 \wedge x_2' = x_2) \\
  x_1 > 0 & \Leftarrow & \inv(d,b,\ve{\progvar}_1,\ve{\progvar}_2) \wedge \sch_\TF(d,b,\ve{\progvar}_1,\ve{\progvar}_2) \wedge x_2 > 0 \\
  x_2 > 0 & \Leftarrow & \inv(d,b,\ve{\progvar}_1,\ve{\progvar}_2) \wedge \sch_\FT(d,b,\ve{\progvar}_1,\ve{\progvar}_2) \wedge x_1 > 0 \\
  \multicolumn{3}{l}{\sch_\TF(d,b,\ve{\progvar}_1,\ve{\progvar}_2) \vee \sch_\FT(d,b,\ve{\progvar}_1,\ve{\progvar}_2) \vee \sch_\TT(d,b,\ve{\progvar}_1,\ve{\progvar}_2) \Leftarrow \inv(d,b,\ve{\progvar}_1,\ve{\progvar}_2) \wedge (x_1 > 0 \vee x_2 > 0)} \\
  -b \leq d \wedge d \leq b \land b \geq 0  & \Leftarrow & \inv(d,b,\ve{\progvar}_1,\ve{\progvar}_2) \wedge x_1 > 0 \wedge x_2 > 0 \\
  \wfr_1(\ve{\progvar}_1,\ve{\progvar}_1') & \Leftarrow & \inv(d,b,\ve{\progvar}_1,\ve{\progvar}_2) \wedge x_2\leq0 \wedge x_1>0 \wedge x_1'=x_1-y_1 \\
  \wfr_2(\ve{\progvar}_2,\ve{\progvar}_2') & \Leftarrow & \inv(d,b,\ve{\progvar}_1,\ve{\progvar}_2) \wedge x_1\leq0 \wedge x_2>0 \wedge x_2'=x_2-2 \cdot y_2
\end{array}
\]
where $\ve{\progvar}_1 = x_1,y_1$, $\ve{\progvar}_1' = x_1',y_1$, $\ve{\progvar}_2 = x_2,y_2$, $\ve{\progvar}_2' = x_2',y_2$, $\TF = \{1\}$, $\FT = \{2\}$, and $\TT = \{1,2\}$.

\section{The \pcsat{} generated solution of Example~\ref{ex:coterm}}
\label{app:sol_coterm}

\[
  \begin{array}{rcl}
  \fnbnd(x_1,y_1,x_2,u_2,b) &\equiv& b = 1\\
  \inv(d, b, x_1, y_1, x_2, y_2) &\equiv&
  d = 0 \land b \geq 0 \land b \leq 1 \land
  \left(\begin{array}{l}
    x_1 = x_2 \land y_1 = y_2\ \lor \\
    y_1 = y_2 \land x_1 + y_1 \geq 1 \land x_2 + 2 \cdot y_2 \geq 1
  \end{array}\right) \\
  \wfr_1(x, y, x', y') &\equiv&
  \left(\begin{array}{l}
    x - 1 \geq 0 \land x - 1 > x' - 1\  \land \\
    ((x' > 0 \land y' \leq 0) \Rightarrow x - 1 \geq 0 \land x - 1 > 1 - y')\ \lor \\
    (x > 0 \land y \leq 0) \land 1 - y \geq 0 \land 1 - y > x' - 1\ \land \\
    ((x' > 0 \land y' \leq 0) \Rightarrow 1 - y \geq 0 \land 1 - y > 1 - y')
  \end{array}\right) \\
  \wfr_2(x, y, x', y') &\equiv&
  \left(\begin{array}{l}
    (y' \geq 0 \lor x' > 0 \land y' \geq 0)\ \land \\
    (y \geq 0 \land (y' \geq 0 \Rightarrow -y \geq 0 \land -y > -y')\ \land \\
    \left(\begin{array}{l}
      (x' > 0 \land y' \geq 0) \Rightarrow -y \geq 0 \land -y > x')\ \lor \\
      (x > 0 \land y \geq 0) \land (y' \geq 0 \Rightarrow x \geq 0 \land x > -y')\ \land \\
      ((x' > 0 \land y' \geq 0) \Rightarrow x \geq 0 \land x > x')
    \end{array}\right)
  \end{array}\right)
\end{array}
\]
Here, $\wfr_1$ is induced by the ranking function $\max(x - 1, \texttt{if } x > 0 \land y \leq 0 \texttt{ then } 1 - y \texttt{ else } -\infty)$ and $\wfr_2$ is induced by the ranking function $\max (\texttt{if } y \geq 0 \texttt{ then } -y \texttt{ else } -\infty, \texttt{if } x > 0 \land y \geq 0 \texttt{ then } x \texttt{ else } -\infty)$.

 \section{The encoding of the TS-GNI verification from Example~\ref{ex:tsgni}}
 \label{app:tsgni}

\[
\begin{array}{rcl}
  \inv(0,b,\ve{\progvar}_1,\ve{\progvar}_2) & \Leftarrow & \fnbnd(x_1,h_1,l_1,x_2,h_2,l_2,b) \wedge b_1 \wedge b_2 \wedge l_1 = l_2\ \wedge \\
  & & (h_1 \wedge x_1 = n_1 \vee \neg h_1 \wedge x_1=l_1)\ \wedge \\
  & & (h_2 \wedge \fnr(p,h_2,l_2,x_2) \vee \neg h_2 \wedge x_2=l_2) \\
  \inv(d',b,\ve{\progvar}_1',\ve{\progvar}_2) & \Leftarrow & \inv(d,b,\ve{\progvar}_1,\ve{\progvar}_2) \wedge \sch_\TF(d,b,\ve{\progvar}_1,\ve{\progvar}_2)\ \wedge \\
  & &
  \left(\begin{array}{l}
    b_1 \wedge h_1 \wedge (x_1 \geq l_1 \wedge \neg b_1' \wedge x_1' = x_1 \vee x_1 < l_1 \wedge b_1' \wedge x_1' = x_1)\ \vee \\
    b_1 \wedge \neg h_1 \wedge (b_1' \wedge x_1' = x_1 + 1 \vee \neg b_1' \wedge x_1' = x_1)\ \vee \\
    \neg b_1 \wedge \neg b_1' \wedge x_1' = x_1
  \end{array}\right)\ \wedge \\
  & & (\neg b_1 \vee \neg b_2 \vee d' = d + 1) \\
\inv(d',b,\ve{\progvar}_1,\ve{\progvar}_2') \vee \inv(d',b,\ve{\progvar}_1,\ve{\progvar}_2'') & \Leftarrow & \inv(d,b,\ve{\progvar}_1,\ve{\progvar}_2) \wedge \sch_\FT(d,b,\ve{\progvar}_1,\ve{\progvar}_2)\ \wedge \\
  & &
  \left(\begin{array}{l}
    b_2 \wedge h_2 \wedge \left(\begin{array}{l}
      x_2 \geq l_2 \wedge \neg b_2' \wedge x_2' = x_2 \wedge \neg b_2'' \wedge x_2'' = x_2 \vee \\
      x_2 < l_2 \wedge b_2' \wedge x_2' = x_2 \wedge b_2'' \wedge x_2' = x_2
    \end{array}\right)\ \vee \\
    b_2 \wedge \neg h_2 \wedge (b_2' \wedge x_2' = x_2 + 1 \wedge \neg b_2'' \wedge x_2'' = x_2)\ \vee \\
    \neg b_2 \wedge \neg b_2' \wedge x_2' = x_2 \wedge \neg b_2'' \wedge x_2'' = x_2
  \end{array}\right)\ \wedge \\
  & & (\neg b_1 \vee \neg b_2 \vee d' = d - 1) \\
  \inv(d,b,\ve{\progvar}_1',\ve{\progvar}_2') \vee \inv(d,b,\ve{\progvar}_1',\ve{\progvar}_2'') & \Leftarrow & \inv(d,b,\ve{\progvar}_1,\ve{\progvar}_2) \wedge \sch_\TT(d,b,\ve{\progvar}_1,\ve{\progvar}_2)\ \wedge \\
  & &
  \left(\begin{array}{l}
    b_1 \wedge h_1 \wedge (x_1 \geq l_1 \wedge \neg b_1' \wedge x_1' = x_1 \vee x_1 < l_1 \wedge b_1' \wedge x_1' = x_1)\ \vee \\
    b_1 \wedge \neg h_1 \wedge (b_1' \wedge x_1' = x_1 + 1 \vee \neg b_1' \wedge x_1' = x_1)\ \vee \\
    \neg b_1 \wedge \neg b_1' \wedge x_1' = x_1
  \end{array}\right)\ \wedge \\
  & &
  \left(\begin{array}{l}
    b_2 \wedge h_2 \wedge \left(\begin{array}{l}
      x_2 \geq l_2 \wedge \neg b_2' \wedge x_2' = x_2 \wedge \neg b_2'' \wedge x_2'' = x_2 \vee \\
      x_2 < l_2 \wedge b_2' \wedge x_2' = x_2 \wedge b_2'' \wedge x_2' = x_2
    \end{array}\right)\ \vee \\
    b_2 \wedge \neg h_2 \wedge (b_2' \wedge x_2' = x_2 + 1 \wedge \neg b_2'' \wedge x_2'' = x_2)\ \vee \\
    \neg b_2 \wedge \neg b_2' \wedge x_2' = x_2 \wedge \neg b_2'' \wedge x_2'' = x_2
  \end{array}\right) \\
  b_1 & \Leftarrow & \inv(d,b,\ve{\progvar}_1,\ve{\progvar}_2) \wedge \sch_\TF(d,b,\ve{\progvar}_1,\ve{\progvar}_2) \wedge b_2 \\
  b_2 & \Leftarrow & \inv(d,b,\ve{\progvar}_1,\ve{\progvar}_2) \wedge \sch_\FT(d,b,\ve{\progvar}_1,\ve{\progvar}_2) \wedge b_1 \\
  \multicolumn{3}{l}{\sch_\TF(d,b,\ve{\progvar}_1,\ve{\progvar}_2) \vee \sch_\FT(d,b,\ve{\progvar}_1,\ve{\progvar}_2) \vee \sch_\TT(d,b,\ve{\progvar}_1,\ve{\progvar}_2) \Leftarrow \inv(d,b,\ve{\progvar}_1,\ve{\progvar}_2) \wedge (b_1 \vee b_2)} \\
  -b \leq d \wedge d \leq b \land b \geq 0  & \Leftarrow & \inv(d,b,\ve{\progvar}_1,\ve{\progvar}_2) \wedge b_1 \wedge b_2 \\
  x_1=x_2 & \Leftarrow & \inv(d,b,\ve{\progvar}_1,\ve{\progvar}_2) \wedge \neg b_1 \wedge \neg b_2 \wedge p=x_1 \\
  \wfr(\ve{\progvar}_2,\ve{\progvar}_2') & \Leftarrow & \inv(d,b,\ve{\progvar}_1,\ve{\progvar}_2) \wedge \neg b_1 \wedge p=x_1 \wedge b_2 \wedge h_2 \wedge x_2 < l_2 \wedge x_2'=x_2
\end{array}
\]
where $\ve{\progvar}_1 = p, b_1, x_1, h_1, l_1$, $\ve{\progvar}_1' = p, b_1', x_1', h_1, l_1$, $\ve{\progvar}_2 = b_2, x_2, h_2, l_2$, $\ve{\progvar}_2' = b_2', x_2', h_2, l_2$, $\ve{\progvar}_2'' = b_2'', x_2'', h_2, l_2$, $\TF = \{1\}$, $\FT = \{2\}$, and $\TT = \{1,2\}$.

\section{The \pcsat{} generated solution of Example~\ref{ex:tsgni}}
\label{app:sol_tsgni}

\[
  \begin{array}{rcl}
  \fnbnd(x_1,y_1,x_2,u_2,b) &\equiv& b = 0\\
  \fnr(p,h_2,l_2,x_2) &\equiv& x_2 = p\\
  \inv(d,b,\ve{\progvar}_1,\ve{\progvar}_2) & \equiv &
  \left(\begin{array}{l}
  \neg h_1 \land \neg h_2 \land
    d = 0 \land b = 0 \land b_2 \land x_2 \geq l_2 \land l_1 = l_2\ \lor \\
  \neg h_1 \land h_2 \land d = 0 \land b \geq 0 \land x_1 \geq l_1 \land p = x_2 \land l_1 = l_2\ \lor \\
  h_1 \land \neg h_2 \land \left(\begin{array}{l}
    d = 0 \land b = 0 \land b_2 \land l_1 = l_2\ \lor \\
    l_1 = x_2 \land p = x_1 \land x_1 \geq l_1 \land d \geq 1 + b \land\ \\
    b = l_2 \land 1 + 2 \cdot x_1 + 2 \cdot x_2 = 0
  \end{array}\right) \lor \\
  h_1 \land h_2 \land \left(\begin{array}{l}
    d = 0 \land b = 0 \land b_1 \land l_1 = l_2 \land x_2 = p\ \lor \\
    \neg b_1 \land x_1 \geq l_1 \land p = x_2 \land l_1 = l_2
  \end{array}\right)
  \end{array}\right) \\
  \wfr(x, h, l, x', h', l') &\equiv&
  \neg h \land l-x \geq 0 \land l-x > l'-x' \lor
  h \land x \geq 0 \land x > x'
\end{array}
\]
where $\ve{\progvar}_1 = p, b_1, x_1, h_1, l_1$ and $\ve{\progvar}_2 = b_2, x_2, h_2, l_2$.

\section{Proof of Theorem~\ref{thm:ksafety}}

\begin{proof}

  \begin{description}[listparindent=\parindent,itemsep=1em]
  \item[(only-if)]\mbox{}\\
  The only-if direction holds from the completeness of the standard
  lock-step product program construction that executes each programs
  synchronously in parallel.  That is, such a product program is
  realized by the the scheduler defined by $\sch_{[k]} = \true$ and
  $\sch_{A} = \false$ for all $A \neq [k]$.  Note that clauses (4) and
  (5) are trivially satisfied by such a scheduler.  The corresponding
  $\inv$ is the set of tuples of states reachable by the lock-step
  evaluation from the tuples of states satisfying $\pre$ (or any
  invariant used to verify the input instance under the lock-step
  product program construction).  It is easy to see that such
  $\inv$ satisfies the rest of the clauses with the scheduler.  (A
  similar argument can also be made with the standard sequential
  product program construction.)
  
  \item[(if)]\mbox{}\\
  We prove the if direction by proving the contrapositive.  So,
  suppose that the tuple of programs violates the $k$-safety property.
  Then, there must be sequences $\pi_1$, \dots, $\pi_k$ such that (a)
  $\pre(\pi_1[1],\dots,\pi_k[1])$ is true, (b) for each $\pi_i$,
  $F_i(\pi_i[\length{\pi_i}])$ is true, (c)
  $\neg \post(\pi_1[\length{\pi_1}],\dots,\pi_k[\length{\pi_k}])$ is true, and (d)
  for each $\pi_i$ and $1 < j \leq \length{\pi_i}$,
  $T_i(\pi_i[j-1],\pi_i[j])$ is true.  The following argument,
  roughly, says that under any scheduler satisfying $\mathcal{C}_\textrm{S}$, we
  can ``reach'' $\pi_1[\length{\pi_1}],\dots,\pi_k[\length{\pi_k}]$ from
  $\pi_1[1],\dots,\pi_k[1]$.  Let $a_i = 1$ for each $i \in [k]$.  Let
  $\ve{v} = \pi_1[a_1],\dots,\pi_k[a_k]$.  By (a) and clause (1), it
  must be the case that $\inv(\ve{v})$ is true.

  If all programs have terminated (i.e., $F_i(\pi_i[a_i])$ for each
  $i$), then by (b), (c) and the fact that programs self-loop after
  reaching final states, $\ve{v}$ invalidates clause (3) and therefore
  we have shown that $\mathcal{C}_\textrm{S}$ is unsatisfiable.  So, suppose that
  there are unfinished programs (i.e.,
  $\bigvee_{i \in [k]} \neg F_i(\pi_i[a_i])$ is true).  By clause (5),
  there must be some $A \in \nepset{[k]}$ such that $\sch_A(\ve{v})$
  is true.  Then, by clause (4), there must be unfinished programs
  that $A$ schedules to be evaluated next, that is,
  $B = \{ i \in A \mid \neg F_{i}(\pi_{i}[a_{i}]) \}$ is non-empty.
  Let $a_i' = a_i + 1$ if $i \in B$ and otherwise let $a_i' = a_i$.
  Let $\ve{v}' = \pi_1[a_1'],\dots,\pi_k[a_k']$.  Then,
  $\inv(\ve{v}')$ must be true by (d) and clause (2).

  Repeating the argument in the above paragraph by setting $\ve{v}$ to
  $\ve{v}'$ and $a_i$ to $a_i'$ for each $i \in [k]$, we will
  eventually reach the terminal tuple of states
  $\pi_1[\length{\pi_1}],\dots,\pi_k[\length{\pi_k}]$ and show that
  $\inv(\pi_1[\length{\pi_1}],\dots,\pi_k[\length{\pi_k}])$ must be true, which,
  by (c), invalidates clause (3) and therefore $\mathcal{C}_\textrm{S}$ is
  unsatisfiable.
  \end{description}
  \qed
\end{proof}

\section{Proof of Theorem~\ref{thm:coterm}}

\begin{proof}
  \begin{description}[listparindent=\parindent,itemsep=1em]
    \item[(only-if)]\mbox{}\\
  Suppose that given
  pair of programs co-terminate.  Let $\sch_\TT = \true$ and $\sch_\FT
  = \sch_\TF = \false$, \ie, let the scheduler be lock-step.  Let
  $\fnbnd(\ve{\progvar},b) = b = 0$ (any other predicate that sets $b$ to be
  non-negative also works).  Let $\inv$ be the set of tuples
  $(d,0,\ve{v_1},\ve{v_2})$ where $\ve{v_1},\ve{v_2}$ are the tuples
  of states reachable from $\pre$ by lock-step evaluation and $d = 0$
  if $\neg F_1(\ve{v_1}) \wedge \neg F_2(\ve{v_2})$ ($d$ is arbitrary
  if $F_1(\ve{v_1}) \vee F_2(\ve{v_2})$).  Let $R$ be the set of
  states of $P_2$ reachable from a pair of states satisfying $\pre$
  after the corresponding execution of $P_1$ has terminated by
  lock-step evaluation.  That is, $R$ is the set of states $\ve{v_2}$
  satisfying the following: there exist a pair of states
  $(\ve{v_1}',\ve{v_2}')$ and a state $\ve{v_1}$ such that
  $\pre(\ve{v_1}',\ve{v_2}')$ is true, $(\ve{v_1},\ve{v_2})$ is
  reachable from $(\ve{v_1}',\ve{v_2}')$ by lock-step evaluation, and
  $F_1(\ve{v_1})$ is true.  Let $\wfr = (R
  \times R) \cap \{ (\ve{v},\ve{v}') \mid T_2(\ve{v},\ve{v}') \}$
  (or any other well-founded relation witnessing the termination of
  $R$).  It is easy to see that these predicates satisfy
  $\mathcal{C}_\textrm{CoT}$.

  \item[(if)]\mbox{}\\
  We prove the if direction by proving the contrapositive.  So,
  suppose that the pair of programs violates co-termination.  Then,
  there must be states $\ve{v_1}$,$\ve{v_2}$, and $\ve{v_1}'$ such that
  $\pre(\ve{v_1},\ve{v_2})$ is true, $\ve{v_1} \reach_1 \ve{v_1}'$, and
  $\ve{v_2} \reach_2 \bot$.  Let $\pi_1$ be a finite sequence such that
  (a1) $\ve{v_1} = \pi_1[1]$, (b1) $F_1(\pi_1[\length{\pi_1}])$ is
  true, and (c1) for each $1 < i \leq \length{\pi_1}$,
  $T_1(\pi_1[i-1],\pi_1[i])$ is true.  Let $\varpi_2$ be an infinite
  sequence such that (a2) $\ve{v_2} = \varpi_2[1]$, (b2) for each $i
  \geq 1$, $F_2(\varpi_2[i])$ is false, and (c2) for each $i > 1$,
  $T_2(\varpi_2[i-1],\varpi_2[i])$ is true.

  Let $c \in \mathbb{Z}$ be such that $\fnbnd(\pi_1[1],\varpi_2[1],c)$
  is true.  By clause (1), it must be the case that
  $\inv(0,c,\pi_1[1],\varpi_2[1])$ is true.  We next show the
  following lemma.
  \begin{lemma}
    \label{lem:wfr}
    Let $d \in \mathbb{Z}$, $1 \leq i \leq \length{\pi_1}$, and $1
    \leq j$.  Suppose $\inv(d,c,\pi_1[i],\varpi_2[j])$ and
    $F_1(\pi_1[i])$.  Then, clauses {\normalfont (6), (5), (4a), (4b),
      (3a), (3c)} cannot be simultaneously satisfied.

  \end{lemma}
  \begin{proof}
    By $F_1(\pi_1[i])$, (4a), (4b), and (5), we have that
    $\sch_\TT(d,c,\pi_1[i],\varpi_2[j])$ or
    $\sch_\FT(d,c,\pi_1[i],\varpi_2[j])$ must be true.  In either
    case, by (c2), (3a), (3c) and the fact that $T_1(\pi_1[i],\pi_1[i])$,
    $\inv(d',b,\pi_1[i],\varpi_2[j+1])$ must be true for some $d' \in
    \mathbb{Z}$ (incidentally, $d' = d+1$ or $d' = d$).  Also, by
    $F_1(\pi_1[i])$ and (6b),
    $\wfr(\varpi_2[j],\varpi_2[j+1])$ is true.

    Repeating the above argument with $j$ updated to $j+1$ and $d$
    updated to $d'$, we derive that
    $\wfr(\varpi_1[j'],\varpi_2[j'+1])$ must be true for all $j' \geq j$.
    However, this violates the condition that $\wfr$ is a
    well-founded relation.  Therefore, the clauses cannot be
    satisfied. \qed
  \end{proof}

  We now return to the proof of the theorem.  Let $d = 0$, $a_1 = 1$
  and $a_2 = 1$.  Note that $\inv(d,c,\pi_1[a_1],\varpi_2[a_2])$ is true
  and $\length{d} \leq c$.
  
  If $F_1(\pi_1[a_1])$ is true, then by Lemma~\ref{lem:wfr}, the constraint
  is unsatisfiable.  So, suppose that $\neg F_1(\pi_1[a_1])$.
  By clause (5), it must be the case that either
  \begin{description}[itemsep=0pt]
  \item[(s$_\TT$)] $\sch_\TT(d,c,\pi_1[a_1],\varpi_2[a_2])$ is true;
  \item[(s$_\TF$)] $\sch_\TF(d,c,\pi_1[a_1],\varpi_2[a_2])$ is true; or
  \item[(s$_\FT$)] $\sch_\FT(d,c,\pi_1[a_1],\varpi_2[a_2])$ is true.
  \end{description}
  If (s$_\TT$) then let $d' = d$, $a_1' = a_1+1$, and $a_2' = a_2+1$.
  If (s$_\FT$) then let $d' = d-1$, $a_1' = a_1$, and $a_2' = a_2+1$.
  If (s$_\TF$ then let $d' = d+1$, $a_1' = a_1$, and $a_2' = a_2+1$.
  In any case, by (b2), (3a), (3b), and (3c),
  $\inv(d',c,\pi_1[a_1'],\varpi_2[a_2'])$ must be true.  But by (b2)
  and (2), it must be the case that $\length{d'} \leq c$.

  Repeating the argument in the above paragraph by setting $d = d'$,
  $a_1 = a_1'$, and $a_2 = a_2'$, we will eventually reach $a_1 =
  \length{\pi_1}$ such that $\inv(d,c,\pi_1[a_1],\varpi_2[a_2])$ is
  true for some $d$ and $a_2$, because $\sch_\FT$ can only become true
  finitely often due to the difference bound.  At this point,
  $F_1(\pi_1[a_1])$ is true.  Therefore by Lemma~\ref{lem:wfr}, the
  constraint is unsatisfiable.
  \end{description}
  \qed
\end{proof}
  
\section{Proof of Theorem~\ref{thm:tigni}}

\begin{proof}

\begin{description}[listparindent=\parindent,itemsep=1em]
  \item[(if)]\mbox{}\\
Suppose that $\rho$ is a solution to $\mathcal{C}_{TIGNI}$.  Let
$P_2'$ be a program whose transition relation $T_2'$ is defined as follows:
  \[
  T_2' = \{ (\ve{v},\ve{v_2},\ve{v},\ve{v_2}') \mid
  \exists r. \rho(\fnr)(\ve{v},\ve{v_2},r) \wedge U_2(r,\ve{v_2},\ve{v_2}') \}
  \]
where $\length{\ve{v}} = \length{\ve{\progvar_1}}$ and
$\fnr$ is the functional predicate variable added in (m6) of
Def.~\ref{def:tigni_encode}.  That is, $P_2'$ is $P_2$ augmented to
(1) take prophecy values as input and propagate them across the
transitions, and (2) determinize the transitions by the assignment to $\fnr$ given in $\rho$.  We write $\reach_2'$ for the
reachability relation of $P_2'$.\footnote{The encoding given in
  Sec.~\ref{sec:appsgni} modified $P_1$ to propagate the prophecy
  values.  Here, we modify $P_2$ for the job.  The resulting
  constraint set is equivalent, but the proof becomes somewhat
  simpler.}

Let $\mathcal{C}_\textrm{TIGNI}'$ be $\mathcal{C}_\textrm{TIGNI}$ but
modified so that its occurrences of $\fnr(\ve{\pfv},\ve{\progvar_2},r)$ are
replaced by $\rho(\fnr)(\ve{\pfv},\ve{\progvar_2},r)$.  Clearly, $\rho$ is a
solution to $\mathcal{C}_\textrm{TIGNI}'$.  Also,
$\mathcal{C}_\textrm{TIGNI}'$ is a $k$-safety constraint set
(cf.~Def.~\ref{def:ksafe_encode}) for $P_1$ and $P_2'$ against the
pre-condition $\pre'$ and the post-condition $\post'$. Therefore, by
Theorem~\ref{thm:ksafety}, $P_1$ and $P_2'$ satisfies the $k$-safety
property given by $\pre'$ and $\post'$.

Now, suppose that $\pre(\ve{v_1},\ve{v_2})$ and $\ve{v_1} \reach_1
\ve{v_1}'$.  Let $\ve{v}$ be valuations of
prophecy variables (\ie, $\length{\ve{v}} = \length{\ve{v_1}}$).  We
have $\pre'(\ve{v},\ve{v_1},\ve{v_2})$.  Also, by $k$-safety of $P_1$
and $P_2'$, either (a) $(\ve{v},\ve{v_2}) \reach_2' \bot$ or (b) there
exists $\ve{v_2}'$ such that $(\ve{v},\ve{v_2}) \reach_2'
(\ve{v},\ve{v_2}')$ and
$\post'(\ve{v},\ve{v_1}',\ve{v_2}')$.

If the former is true for some $\ve{v}$, then we have $\ve{v_2}
\reach_2 \bot$ by letting $P_2$ resolve the non-deterministic choices
according to the execution $(\ve{v},\ve{v_2}) \reach_2' \bot$.
Otherwise, (b) holds for all $\ve{v}$.  Therefore, there exists
$\ve{v_2}'$ such that $(\ve{v_1},\ve{v_2}) \reach_2'
(\ve{v_1},\ve{v_2}')$ and $\post'(\ve{v_1},\ve{v_1}',\ve{v_2}')$.
Therefore, by letting $P_2$ resolve the non-deterministic choices
according to the execution $(\ve{v_1},\ve{v_2}) \reach_2'
(\ve{v_1},\ve{v_2}')$, we have $\ve{v_2} \reach_2 \ve{v_2}'$ and
$\post(\ve{v_1}',\ve{v_2}')$.  Therefore, $P_1$ and $P_2$ satisfies
TI-GNI given by $\pre$ and $\post$.

\item[(only-if)]\mbox{}\\
  Suppose that $P_1$ and $P_2$ satisfy TI-GNI given by
  $\pre$ and $\post$.  Let us write $\pi : \ve{v} \reach_i \ve{v}'$ if
  $\pi$ is a finite sequence witnessing the reachability relation
  $\ve{v} \reach_i \ve{v}'$.  Likewise, let us write $\varpi : \ve{v}
  \reach_i \bot$ if $\varpi$ is an infinite sequence witnessing the
  non-termination $\ve{v} \reach_i \bot$.
  
  For each state $\ve{v_1}'$ of $P_1$, we define $L(\ve{v_1}')$ to be
  a (possibly infinite) list of finite and infinite sequences obtained
  by totally ordering the elements of the following set:
  \[
  \begin{array}{l}
    \{ \pi \mid \exists \ve{v_1},\ve{v_2},\ve{v_1}'\ve{v_2}'.\pre(\ve{v_1},\ve{v_2}) \wedge \ve{v_1} \reach_1 \ve{v_1}' \wedge \pi:\ve{v_2} \reach_2 \ve{v_2}' \wedge \post(\ve{v_1}',\ve{v_2}') \} \ \cup \\
    \{ \varpi \mid \exists \ve{v_1},\ve{v_2},\ve{v_1}'.\pre(\ve{v_1},\ve{v_2}) \wedge \ve{v_1} \reach_1 \ve{v_1}' \wedge \varpi:\ve{v_2} \reach_2 \bot \}
    \end{array}
  \]
  That is, $L(\ve{v_1}')$ is a list of the terminating and
  non-terminating execution traces of $P_2$ that can match a
  (terminating) execution trace of $P_1$ whose final state is
  $\ve{v_1}'$.  Note that, if $\pre(\ve{v_1},\ve{v_2})$ and $v_1
  \reach_1 \ve{v_1}'$, then there exists $\xi \in L(\ve{v_1}')$ such
  that $\xi[1] = \ve{v_2}$.

  For $\xi \in L(\ve{v_1}')$ and $1 \leq i \leq \length{\xi}$ (where
  $\length{\varpi} = \infty$ for an infinite trace $\varpi$), let
  $\angelchoice(\xi,i)$ be the angelic non-deterministic choice made
  in the $i$-th step of the execution $\xi$.  For a function $f$, we
  write $f[a \mapsto b]$ for the function defined by $f[a \mapsto
    b](a) = b$ and $f[a \mapsto b](c) = f(c)$ for all $c \neq a$.

  Next, we define a function $\detfun$ by the following process. Initialize
  $\detfun \leftarrow \emptyset$.  Then, for each state $\ve{v_1}'$ of
  $P_1$, apply the steps below until $L(\ve{v_1}')$ is empty:
  
  \begin{itemize}[itemsep=0pt]
  \item[1.] Take the head element $\xi \in L(\ve{v_1}')$.
  \item[2.] Scan $\xi$ forwards and, at each position $i$, record its angelic choice by updating $\detfun \leftarrow \detfun[(\ve{v_1}',\xi[i]) \mapsto \angelchoice(\xi,i)]$ if $(\ve{v_1}',\xi[i]) \notin \dom{\detfun}$ and otherwise leave $\detfun$ unchanged.
  \end{itemize}
  Finally, for each pair $(\ve{v_1},\ve{v_2})$ of states of $P_1$ and
  and $P_2$ such that $(\ve{v_1},\ve{v_2}) \notin \dom{\detfun}$,
  update $\detfun$ by setting $\detfun \leftarrow
  \detfun[(\ve{v_1},\ve{v_2}) \mapsto r]$ where $r$ is arbitrary.
  Note that this is an infinite ``process'' in general since the
  number of states of $P_1$, $\length{L(\ve{v_1}')}$, and the length
  of a sequence in $L(\ve{v_1}')$, can all be infinite.  However, it
  is well-defined.  Importantly, $\detfun$ thus constructed is a total
  function from the pairs of $P_1$ and $P_2$ states.

  Note that using $\detfun$ as the determinizing choice function in
  $P_2'$ would make $P_1$ and $P_2'$ (cf.~the soundness proof above)
  satisfy the $k$-safety property given by $\pre'$ and $\post'$.  This
  follows from the fact that, if $\pre(\ve{v_1},\ve{v_2})$ and
  $\ve{v_1} \reach_1 \ve{v_1}'$, then any execution of $P_2'$ from
  $(\ve{v_1}',\ve{v_2})$ only visits states $(\ve{v_1}',\ve{v})$ where
  $\ve{v}$ occurs in $L(\ve{v_1}')$ and therefore may only reach an
  output $(\ve{v_1}',\ve{v_2}')$ such that
  $\post(\ve{v_1}',\ve{v_2}')$ is true.\footnote{However, the behavior
    of $P_2'$ is not necessarily equivalent to what the traces in
    $L(\ve{v_1}')$ stipulate.  For instance, $P_2'$ may non-terminate
    even when all traces in $L(\ve{v_1}')$ are finite.}
  
  Therefore, the rest of the proof follows the structure of the
  completeness direction of the proof of Theorem~\ref{thm:ksafety}.
  Let us call a pair of states $\ve{v_1}$ and
  $(\ve{v_1}',\ve{v_2})$ of $P_1$ and $P_2'$ {\em initial} if
  $\pre'(\ve{v_1}',\ve{v_1},\ve{v_2})$ is true (\ie,
  $\pre(\ve{v_1},\ve{v_2})$ is true).  Also, in what follows, we
  assume that $P_2'$ uses $\detfun$ as the determinizing choice function.

  Let $\sch_\TT = \true$ and $\sch_\FT = \sch_\TF = \false$, \ie, let
  the scheduler be lock-step.  Let $\fnr$ be the set of tuples
  $(\ve{v_1}',\ve{v_2},r)$ such that $\detfun(\ve{v_1}',\ve{v_2}) = r$, \ie,
  $\fnr$ expresses the graph of $\detfun$.  Note that $\fnr$ trivially
  satisfies the function-ness requirement.  Let $\inv$ be the set of
  tuples $(\ve{v_1}',\ve{v_1},\ve{v_2})$ where $\ve{v_1}$ and
  $(\ve{v_1}',\ve{v_2})$ are pair of states of $P_1$ and $P_2'$
  reachable from some initial pair of states by lock-step evaluation.
  It is easy to see that these predicates satisfy $\mathcal{C}_\textrm{TIGNI}$.
\end{description}
\qed
\end{proof}

\section{Proof of Theorem~\ref{thm:tsgni}}

\begin{proof}

\begin{description}[listparindent=\parindent,itemsep=1em]
\item[(if)]\mbox{}\\

  The proof is similar to the soundness direction of
  Theorem~\ref{thm:tigni} and proceeds by determinizing $P_2$ using
  the solution to the constraint set.  So, suppose that $\rho$ is a
  solution to $\mathcal{C}_{TSGNI}$.  Let $P_2'$ be a program whose
  transition relation $T_2'$ is defined as follows:
  \[
  T_2' = \{ (\ve{v},\ve{v_2},\ve{v},\ve{v_2}') \mid
  \exists r. \rho(\fnr)(\ve{v},\ve{v_2},r) \wedge U_2(r,\ve{v_2},\ve{v_2}') \}
  \]
  where $\length{\ve{v}} = \length{\ve{\progvar_1}}$ and $\fnr$ is the
  functional predicate variable added in (m6) of
  Def.~\ref{def:tigni_encode}.  That is, $P_2'$ is $P_2$ augmented to
  (1) take prophecy values as input and propagate them across the
  transitions, and (2) determinize the transitions by the assignment
  to $\fnr$ given in $\rho$.  We write $\reach_2'$ for the
  reachability relation of $P_2'$.  We also consider modification of
  $P_1$, $P_1'$, that is defined by the transition relation $T_1'$ and
  the final state predicate $F_1'$ given in
  Def.~\ref{def:tsgni_encode}.  We write $\reach_1'$ for the
  reachability relation of $P_1'$.  Note that, propagating the
  prophecy values in both $P_1'$ and $P_2'$ is equivalent to
  propagating them in one of the two, since neither transitions modify
  the prophecy values.

  Let $\mathcal{C}_\textrm{TSGNI}'$ be $\mathcal{C}_\textrm{TSGNI}$
  but modified so that its occurrences of $\fnr(\ve{\pfv},\ve{\progvar_2},r)$
  are replaced by $\rho(\fnr)(\ve{\pfv},\ve{\progvar_2},r)$.  Clearly, $\rho$
  is a solution to $\mathcal{C}_\textrm{TSGNI}'$.  Note that
  $\mathcal{C}_\textrm{TSGNI}'$ asserts a conjunction of the
  $k$-safety for $P_1'$ and $P_2'$ against the pre-condition $\pre'$
  and the post-condition $\post'$ and the co-termination for $P_1'$
  and $P_2'$ against the pre-condition $\pre'$.  Therefore, by
  Theorem~\ref{thm:ksafety}, (a) $P_1'$ and $P_2'$ satisfies the
  $k$-safety property given by $\pre'$ and $\post'$, and by
  Theorem~\ref{thm:coterm}, (b) $P_1'$ and $P_2'$ satisfies the
  co-termination property given by $\pre'$.

  Now, suppose that $\pre(\ve{v_1},\ve{v_2})$ and $\ve{v_1} \reach_1
  \ve{v_1}'$.  We have $\pre'(\ve{v_1},\ve{v_1},\ve{v_2})$ and
  $(\ve{v_1},\ve{v_1}) \reach_1' (\ve{v_1},\ve{v_1}')$.  Therefore, by
  (b), we have that $(\ve{v_1},\ve{v_2}) \not \reach_2' \bot$.
  Therefore, $(\ve{v_1},\ve{v_2}) \reach_2' (\ve{v_1},\ve{v_2}')$ for
  some $\ve{v_2}'$, assuming that every pre-related state has at least
  one execution.  Then, by (a), it must be the case that
  $\post'(\ve{v_1},\ve{v_1}',\ve{v_2}')$.  Therefore, by letting $P_2$
  resolve the non-deterministic choices according to the execution
  $(\ve{v_1},\ve{v_2}) \reach_2' (\ve{v_1},\ve{v_2}')$, we have
  $\ve{v_2} \reach_2 \ve{v_2}'$ and $\post(\ve{v_1}',\ve{v_2}')$.
  Therefore, $P_1$ and $P_2$ satisfies TS-GNI given by $\pre$ and
  $\post$.

\item[(only-if)]\mbox{}\\
  The proof is similar to the completeness
  direction of Theorem~\ref{thm:tigni} and proceeds by constructing a
  determinizing choice function $F$ from the executions of the two
  programs.  We will use some of the notations defined there.  That
  is, we write $\pi : \ve{v} \reach_i \ve{v}'$ if $\pi$ is a finite
  sequence witnessing the reachability relation $\ve{v} \reach_i
  \ve{v}'$, and when $\pi$ is an execution of $P_2$ and $1 \leq i \leq
  \length{\pi}$, we write $\angelchoice(\pi,i)$ be the angelic
  non-deterministic choice made in the $i$-th step of $\pi$.  Also,
  for a function $f$, we write $f[a \mapsto b]$ for the function
  defined by $f[a \mapsto b](a) = b$ and $f[a \mapsto b](c) = f(c)$
  for all $c \neq a$.

  Now we proceed with the proof of the completeness direction.  So,
  suppose that $P_1$ and $P_2$ satisfy TI-GNI given by $\pre$ and
  $\post$.  For each state $\ve{v_1}'$ of $P_1$, we define
  $L(\ve{v_1}')$ to be a (possibly infinite) list of {\em finite}
  sequences obtained by totally ordering the elements of the following
  set:
  \[
  \{ \pi \mid \exists \ve{v_1},\ve{v_2},\ve{v_1}'\ve{v_2}'.\pre(\ve{v_1},\ve{v_2}) \wedge \ve{v_1} \reach_1 \ve{v_1}' \wedge \pi:\ve{v_2} \reach_2 \ve{v_2}' \wedge \post(\ve{v_1}',\ve{v_2}') \}
  \]
  That is, $L(\ve{v_1}')$ is a list of the terminating execution traces of
  $P_2$ that can match a (terminating) execution trace of $P_1$ whose
  final state is $\ve{v_1}'$.  Note that, if $\pre(\ve{v_1},\ve{v_2})$
  and $v_1 \reach_1 \ve{v_1}'$, then there exists $\pi \in
  L(\ve{v_1}')$ such that $\pi[1] = \ve{v_2}$.

  Next, we define a function $\detfun$ by the following process. Initialize
  $\detfun \leftarrow \emptyset$.  Then, for each state $\ve{v_1}'$ of
  $P_1$, apply the steps below until $L(\ve{v_1}')$ is empty:

  \begin{itemize}[itemsep=0pt]
  \item[1.] Take the head element $\pi \in L(\ve{v_1}')$.
  \item[2.] Scan $\pi$ forwards and, at each position $i$, record its
    angelic choice by updating $\detfun \leftarrow
    \detfun[(\ve{v_1}',\pi[i]) \mapsto \angelchoice(\pi,i)]$.
  \end{itemize}
  Finally, for each pair $(\ve{v_1},\ve{v_2})$ of states of $P_1$ and
  and $P_2$ such that $(\ve{v_1},\ve{v_2}) \notin \dom{\detfun}$,
  update $\detfun$ by setting $\detfun \leftarrow
  \detfun[(\ve{v_1},\ve{v_2}) \mapsto r]$ where $r$ is arbitrary.  As
  that in the proof of Theorem~\ref{thm:tigni}, this is an infinite
  ``process'' in general since the number of states of $P_1$ and
  $\length{L(\ve{v_1}')}$ can both be infinite.  However, it is
  well-defined.  Importantly, $\detfun$ thus constructed is a total
  function from the pairs of $P_1$ and $P_2$ states.

  An important difference from the construction of $\detfun$ given in
  Theorem~\ref{thm:tigni} is that in step 2., we always update $\detfun$ by
  the angelic choice, \ie, even if $(\ve{v_1}',\xi[i]) \in \dom{\detfun}$.
  This ensures the ``co-termination'' of $P_2'$ when it is
  determinized by $\detfun$.\footnote{The above ``overwriting''
    construction of $\detfun$ actually also works for the proof of
    Theorem~\ref{thm:tigni}.  However, the construction method given
    there does not work here because it may introduce unwanted
    non-termination.}
  
  Note that using $\detfun$ as the determinizing choice function in $P_2'$
  would make $P_1'$ and $P_2'$ (cf.~the soundness proof above) satisfy
  the $k$-safety property given by $\pre'$ and $\post'$ and the
  co-termination property given by $\pre'$.  This follows from the
  fact that, if $\pre(\ve{v_1},\ve{v_2})$ and $\ve{v_1} \reach_1
  \ve{v_1}'$, then any execution of $P_2'$ from $(\ve{v_1}',\ve{v_2})$
  only visits states $(\ve{v_1}',\ve{v})$ where $\ve{v}$ occurs in
  $L(\ve{v_1}')$ and reaches an output $(\ve{v_1}',\ve{v_2}')$ such
  that $\post(\ve{v_1}',\ve{v_2}')$ is true.  In particular, the
  termination of $P_2'$ from such a state $(\ve{v_1}',\ve{v_2})$ is
  guaranteed by the fact that all traces in $L(\ve{v_1}')$ are finite.

  Therefore, the rest of the proof follows the structures of the
  completeness directions of the proofs of Theorems~\ref{thm:ksafety}
  and \ref{thm:coterm}.  Let us call a pair of states
  $(\ve{v_1}',\ve{v_1})$ and $(\ve{v_1}',\ve{v_2})$ of $P_1'$ and
  $P_2'$ {\em initial} if $\pre'(\ve{v_1}',\ve{v_1},\ve{v_2})$ is
  true (\ie, $\pre(\ve{v_1},\ve{v_2})$ is true).  Also, in what
  follows, we assume that $P_2'$ uses $\detfun$ as the determinizing
  choice function.

  Let $\sch_\TT = \true$ and $\sch_\FT = \sch_\TF = \false$, \ie, let
  the scheduler be lock-step.  Let $\fnr$ be the set of tuples
  $(\ve{v_1}',\ve{v_2},r)$ such that $F(\ve{v_1}',\ve{v_2}) = r$, \ie,
  $\fnr$ expresses the graph of $F$.  Note that $\fnr$ trivially
  satisfies the function-ness requirement.  Let $\fnbnd(\ve{\progvar},b) = b
  = 0$ (any other predicate that sets $b$ to be non-negative also
  works).  Let $\inv$ be the set of tuples
  $(d,0,\ve{v_1}',\ve{v_1},\ve{v_2})$ where $\ve{v_1}',\ve{v_1}$ and
  $(\ve{v_1}',\ve{v_2})$ are reachable from some initial pair of
  states by lock-step evaluation, and $d = 0$ if $\neg
  F_1'(\ve{v_1}',\ve{v_1}) \wedge \neg F_2'(\ve{v_1}',\ve{v_2})$ ($d$
  is arbitrary if $F_1'(\ve{v_1}',\ve{v_1}) \vee
  F_2'(\ve{v_1}',\ve{v_2})$).  Let $R$ be the set of states of $P_2'$
  reachable from some initial pair of states after the corresponding
  execution of $P_1'$ has terminated by lock-step evaluation.  That
  is, $R$ is the set of states $(\ve{v_1}',\ve{v_2})$ satisfying the
  following: there exist a state $(\ve{v_1}',\ve{v_1})$ such that
  $((\ve{v_1}',\ve{v_1}),(\ve{v_1}',\ve{v_2}))$ is reachable from some
  initial pair of states by lock-step evaluation and
  $F_1'(\ve{v_1}',\ve{v_1})$ is true.  Let $\wfr = (R \times R) \cap
  \{ ((\ve{v_1}',\ve{v_2}),(\ve{v_1}',\ve{v_2}')) \mid
  T_2'((\ve{v_1}',\ve{v_2}),(\ve{v_1}',\ve{v_2}')) \}$ (or any other
  well-founded relation witnessing the termination of $R$).  It is
  easy to see that these predicates satisfy
  $\mathcal{C}_\textrm{TSGNI}$.
  \end{description}
\qed
\end{proof}

\section{Unsatisfiability Checking of Example Instances}
\label{sec:unsat_ex}
The unsatisfiability of the given example instances $(\examples,\kind)$ can be decided by an off-the-shelf SAT solver if $\examples$ has only ordinary predicate variables
because $\examples$ is a finite set of clauses not containing term variables.  Otherwise, we use the following (CDCL-like) iterative algorithm staring from $\examples_0=\examples$: For each iteration $i \geq 0$, we first check whether $(\examples_i,\emptyset)$ is unsatisfiable.  If so, then we conclude that $(\examples,\kind)$ is unsatisfiable.  Otherwise, we obtain a satisfying assignment $\sigma$ for $\examples_i$.  Then, for each well-founded predicate variable $X$, we consider the graph comprising the edges $\myset{(\myseq{v}_1,\myseq{v}_2)}{\;\models \sigma(X(\myseq{v}_1,\myseq{v}_2))}$ and enumerate its simple cycles (e.g., by using the algorithm of \cite{Johnson1975}).  Note that such cycles would be counterexamples to the well-foundedness constraint $X$.
Also, for each functional predicate variable $X$, we enumerate the pairs $\finset{X(\myseq{v},v_1),X(\myseq{v},v_2)}$ such that $v_1 \neq v_2$ and $\models \sigma(X(\myseq{v},v_1) \wedge X(\myseq{v},v_2))$, which would be counterexamples to the functionality constraint $X$.
If no such cycles nor pairs exist, we conclude that $(\examples,\kind)$ is satisfiable.  Otherwise, we let $\examples_{i+1}$ be $\examples_i$ but with the following new learnt clauses added:
\begin{itemize}
    \item $\neg X(\myseq{v}_1,\myseq{v}_2) \lor \dots \lor \neg X(\myseq{v}_{m-1},\myseq{v}_m)$ for each simple cycle $\myseq{v}_1,\dots,\myseq{v}_m=\myseq{v}_1$ of each well-founded predicate $X$, and
    \item $\neg X(\myseq{v},v_1) \lor \neg X(\myseq{v},v_2)$ for each pair $\finset{X(\myseq{v},v_1),X(\myseq{v},v_2)}$ of each functional predicate $X$.
\end{itemize}
We then proceed to the next iteration with $\examples_{i+1}$.

It is worth mentioning here that if the original \PCSPWFFN{} $(\clauses,\kind)$ is unsatisfiable and $\clauses$ has no well-founded predicate variable, there always exists an unsatisfiable finite set $\examples$ of example instances of $\clauses$.
However, there is, in general, no such finite witness of the unsatifiability if $\clauses$ has a well-founded predicate variable.

\section{A Refined Stratified Template Family for Well-Founded Predicates}
\label{sec:wftemp}
The stratified template family $\template_X^\KWF$ for well-founded predicates shown in Fig.~\ref{fig:templates} can be further refined without loss of generality by simplifying and using $\IF r(\myseq{x})\geq 0 \THEN r(\myseq{x}) \ELSE {-1}$
instead of $r(\myseq{x})$.

\begin{align*}
\begin{array}{rcl}
\template_X^\Downarrow(\mathit{np},\mathit{nl},\mathit{nc},\mathit{rc},\mathit{rd},\mathit{dc},\mathit{dd})
&\defeq&
\lambda (\myseq{x},\myseq{y}).
\left( \bigvee_{j=1}^{\mathit{np}} D_j(\myseq{y}) \right) \\
&& \;\;\; \land
\left(
\bigvee_{i=1}^{\mathit{np}}
D_i(\myseq{x}) \land
\bigwedge_{j=1}^{\mathit{np}}
\left(D_j(\myseq{y}) \imply
\mathit{DEC}_{i,j}(\myseq{x},\myseq{y})
\right)
\right) \\
\mathit{DEC}_{i,j}(\myseq{x},\myseq{y})
&\defeq&
\bigvee_{k=1}^{\mathit{nl}}
\left(
\begin{array}{l}
r_{i,k}(\myseq{x}) \geq 0 \land
r_{i,k}(\myseq{x}) > r_{j,k}(\myseq{y}) \land \\
\bigwedge_{\ell=1}^{k-1}
\left(
r_{i,\ell}(\myseq{x}) < 0 \land r_{j,\ell}(\myseq{y}) < 0
\lor
r_{i,\ell}(\myseq{x}) \geq r_{j,\ell}(\myseq{y})
\right)
\end{array}
\right)
\end{array}
\end{align*}

\section{Relational Verification Benchmarks}
\label{sec:benchmarks}

Our relational verification benchmark set consists of:
\begin{itemize}
  \item The $k$-safety verification problem \verb|DoubleSquareNI_h**| from Example~\ref{ex:ksafety}, which is originally introduced in \cite{Shemer2019}.
  \item The $k$-safety verification problem \verb|HalfSquareNI| of the following program obtained from \cite{Shemer2019}:
  \begin{alltt}
    pre(low1 == low2)
    halfSquare(int h, int low) \{
      assume(low > h > 0);
      int i = 0, y = 0, v = 0;
      while (h > i) \{
        i++; y += y;
      \}
      v = 1;
      while (low > i) \{
        i++; y += y;
      \}
      return y;
    \}
    post(y1 == y2)
  \end{alltt}
  The encoded constraints are:
  \begin{alltt}
  Inv(b1 : bool, h1, low1, i1, y1, v1, b2 : bool, h2, low2, i2, y2, v2) :-
    low1 = low2, low1 > h1, h1 > 0, low2 > h2, h2 > 0,
    b1, i1 = 0, y1 = 0, v1 = 0,
    b2, i2 = 0, y2 = 0, v2 = 0.
  
  Inv(b1' : bool, h1, low1, i1', y1', v1', b2 : bool, h2, low2, i2, y2, v2) :-
    Inv(b1 : bool, h1, low1, i1, y1, v1, b2 : bool, h2, low2, i2, y2, v2),
    SchTF(b1 : bool, h1, low1, i1, y1, v1, b2 : bool, h2, low2, i2, y2, v2),
    b1 and h1 > i1 and b1' and i1' = i1 + 1 and y1' = y1 + y1 and v1' = v1 or
    b1 and h1 <= i1 and !b1' and i1' = i1 and y1' = y1 and v1' = 1 or
    !b1 and low1 > i1 and !b1' and i1' = i1 + 1 and y1' = y1 + y1 and v1' = v1 or
    !b1 and low1 <= i1 and !b1' and i1' = i1 and y1' = y1 and v1' = v1.
  Inv(b1 : bool, h1, low1, i1, y1, v1, b2' : bool, h2, low2, i2', y2', v2') :-
    Inv(b1 : bool, h1, low1, i1, y1, v1, b2 : bool, h2, low2, i2, y2, v2),
    SchFT(b1 : bool, h1, low1, i1, y1, v1, b2 : bool, h2, low2, i2, y2, v2),
    b2 and h2 > i2 and b2' and i2' = i2 + 1 and y2' = y2 + y2 and v2' = v2 or
    b2 and h2 <= i2 and !b2' and i2' = i2 and y2' = y2 and v2' = 1 or
    !b2 and low2 > i2 and !b2' and i2' = i2 + 1 and y2' = y2 + y2 and v2' = v2 or
    !b2 and low2 <= i2 and !b2' and i2' = i2 and y2' = y2 and v2' = v2.
  Inv(b1' : bool, h1, low1, i1', y1', v1', b2' : bool, h2, low2, i2', y2', v2') :-
    Inv(b1 : bool, h1, low1, i1, y1, v1, b2 : bool, h2, low2, i2, y2, v2),
    SchTT(b1 : bool, h1, low1, i1, y1, v1, b2 : bool, h2, low2, i2, y2, v2),
    b1 and h1 > i1 and b1' and i1' = i1 + 1 and y1' = y1 + y1 and v1' = v1 or
    b1 and h1 <= i1 and !b1' and i1' = i1 and y1' = y1 and v1' = 1 or
    !b1 and low1 > i1 and !b1' and i1' = i1 + 1 and y1' = y1 + y1 and v1' = v1 or
    !b1 and low1 <= i1 and !b1' and i1' = i1 and y1' = y1 and v1' = v1,
    b2 and h2 > i2 and b2' and i2' = i2 + 1 and y2' = y2 + y2 and v2' = v2 or
    b2 and h2 <= i2 and !b2' and i2' = i2 and y2' = y2 and v2' = 1 or
    !b2 and low2 > i2 and !b2' and i2' = i2 + 1 and y2' = y2 + y2 and v2' = v2 or
    !b2 and low2 <= i2 and !b2' and i2' = i2 and y2' = y2 and v2' = v2.
  
  b1 or low1 > i1 :-
    Inv(b1 : bool, h1, low1, i1, y1, v1, b2 : bool, h2, low2, i2, y2, v2),
    SchTF(b1 : bool, h1, low1, i1, y1, v1, b2 : bool, h2, low2, i2, y2, v2),
    b2 or low2 > i2.
  b2 or low2 > i2 :-
    Inv(b1 : bool, h1, low1, i1, y1, v1, b2 : bool, h2, low2, i2, y2, v2),
    SchFT(b1 : bool, h1, low1, i1, y1, v1, b2 : bool, h2, low2, i2, y2, v2),
    b1 or low1 > i1.
  SchTF(b1 : bool, h1, low1, i1, y1, v1, b2 : bool, h2, low2, i2, y2, v2),
  SchFT(b1 : bool, h1, low1, i1, y1, v1, b2 : bool, h2, low2, i2, y2, v2),
  SchTT(b1 : bool, h1, low1, i1, y1, v1, b2 : bool, h2, low2, i2, y2, v2) :-
    Inv(b1 : bool, h1, low1, i1, y1, v1, b2 : bool, h2, low2, i2, y2, v2),
    b1 or low1 > i1 or b2 or low2 > i2.
  
  y1 = y2 :-
    Inv(b1 : bool, h1, low1, i1, y1, v1, b2 : bool, h2, low2, i2, y2, v2),
    !b1, low1 <= i1, !b2, low2 <= i2.
  \end{alltt}
  \item The $k$-safety verification problem \verb|ArrayInsert| of the following program obtained from \cite{Shemer2019} by simulating the control flow of the original array-manipulating program:
  \begin{alltt}
    pre(len1 == len2)
    int arrayInsert(int len, int h) \{
      int i=0;
      while (i < len && i != h) i++;
      len = len + 1;
      while (i < len) i++;
      return i;
    \}
    post(i1 == i2)
  \end{alltt}
  The encoded constraints are:
  \begin{alltt}
  Inv(b1 : bool, len1, h1, i1, b2 : bool, len2, h2, i2) :-
    len1 = len2, b1, i1 = 0, b2, i2 = 0.
  
  Inv(b1' : bool, len1', h1, i1', b2 : bool, len2, h2, i2) :-
    Inv(b1 : bool, len1, h1, i1, b2 : bool, len2, h2, i2),
    SchTF(b1 : bool, len1, h1, i1, b2 : bool, len2, h2, i2),
    b1 and i1 < len1 and i1 <> h1 and b1' and len1' = len1 and i1' = i1 + 1 or
    b1 and (i1 >= len1 or i1 = h1) and !b1' and len1' = len1 + 1 and i1' = i1 or
    !b1 and i1 < len1 and !b1' and len1' = len1 and i1' = i1 + 1 or
    !b1 and i1 >= len1 and !b1' and len1' = len1 and i1' = i1.
  Inv(b1 : bool, len1, h1, i1, b2' : bool, len2', h2, i2') :-
    Inv(b1 : bool, len1, h1, i1, b2 : bool, len2, h2, i2),
    SchFT(b1 : bool, len1, h1, i1, b2 : bool, len2, h2, i2),
    b2 and i2 < len2 and i2 <> h2 and b2' and len2' = len2 and i2' = i2 + 1 or
    b2 and (i2 >= len2 or i2 = h2) and !b2' and len2' = len2 + 1 and i2' = i2 or
    !b2 and i2 < len2 and !b2' and len2' = len2 and i2' = i2 + 1 or
    !b2 and i2 >= len2 and !b2' and len2' = len2 and i2' = i2.
  Inv(b1' : bool, len1', h1, i1', b2' : bool, len2', h2, i2') :-
    Inv(b1 : bool, len1, h1, i1, b2 : bool, len2, h2, i2),
    SchTT(b1 : bool, len1, h1, i1, b2 : bool, len2, h2, i2),
    b1 and i1 < len1 and i1 <> h1 and b1' and len1' = len1 and i1' = i1 + 1 or
    b1 and (i1 >= len1 or i1 = h1) and !b1' and len1' = len1 + 1 and i1' = i1 or
    !b1 and i1 < len1 and !b1' and len1' = len1 and i1' = i1 + 1 or
    !b1 and i1 >= len1 and !b1' and len1' = len1 and i1' = i1,
    b2 and i2 < len2 and i2 <> h2 and b2' and len2' = len2 and i2' = i2 + 1 or
    b2 and (i2 >= len2 or i2 = h2) and !b2' and len2' = len2 + 1 and i2' = i2 or
    !b2 and i2 < len2 and !b2' and len2' = len2 and i2' = i2 + 1 or
    !b2 and i2 >= len2 and !b2' and len2' = len2 and i2' = i2.
  
  b1 or i1 < len1 :-
    Inv(b1 : bool, len1, h1, i1, b2 : bool, len2, h2, i2),
    SchTF(b1 : bool, len1, h1, i1, b2 : bool, len2, h2, i2),
    b2 or i2 < len2.
  b2 or i2 < len2 :-
    Inv(b1 : bool, len1, h1, i1, b2 : bool, len2, h2, i2),
    SchFT(b1 : bool, len1, h1, i1, b2 : bool, len2, h2, i2),
    b1 or i1 < len1.
  SchTF(b1 : bool, len1, h1, i1, b2 : bool, len2, h2, i2),
  SchFT(b1 : bool, len1, h1, i1, b2 : bool, len2, h2, i2),
  SchTT(b1 : bool, len1, h1, i1, b2 : bool, len2, h2, i2) :-
    Inv(b1 : bool, len1, h1, i1, b2 : bool, len2, h2, i2),
    b1 or i1 < len1 or b2 or i2 < len2.
  
  i1 = i2 :-
    Inv(b1 : bool, len1, h1, i1, b2 : bool, len2, h2, i2),
    !b1, i1 >= len1, !b2, i2 >= len2.
  \end{alltt}
  \item The $k$-safety verification problem \verb|SquareSum| of the following program obtained from \cite{Shemer2019} by replacing the nonlinear expression \verb|a*a| in the original program with \verb|a|:
  \begin{alltt}
    pre(a1 < a2 && b2 < b1)
    squaresSum(int a, int b) \{
      assume(0 < a < b);
      int c=0;
      while (a<b) \{ c+=a; a++; \}
      return c;
    \}
    post(c2 < c1)
  \end{alltt}
  The encoded constraints are:
  \begin{alltt}
  Inv(a1, b1, c1, a2, b2, c2) :-
    a1 < a2, b2 < b1,
    0 < a1, a1 < b1, 0 < a2, a2 < b2,
    c1 = 0, c2 = 0.
  
  Inv(a1', b1, c1', a2, b2, c2) :-
    Inv(a1, b1, c1, a2, b2, c2),
    SchTF(a1, b1, c1, a2, b2, c2),
    a1 < b1 and c1' = c1 + a1 and a1' = a1 + 1 or a1 >= b1 and c1' = c1 and a1' = a1.
  Inv(a1, b1, c1, a2', b2, c2') :-
    Inv(a1, b1, c1, a2, b2, c2),
    SchFT(a1, b1, c1, a2, b2, c2),
    a2 < b2 and c2' = c2 + a2 and a2' = a2 + 1 or a2 >= b2 and c2' = c2 and a2' = a2.
  Inv(a1', b1, c1', a2', b2, c2') :-
    Inv(a1, b1, c1, a2, b2, c2),
    SchTT(a1, b1, c1, a2, b2, c2),
    a1 < b1 and c1' = c1 + a1 and a1' = a1 + 1 or a1 >= b1 and c1' = c1 and a1' = a1,
    a2 < b2 and c2' = c2 + a2 and a2' = a2 + 1 or a2 >= b2 and c2' = c2 and a2' = a2.
  
  a1 < b1 :-
    Inv(a1, b1, c1, a2, b2, c2),
    SchTF(a1, b1, c1, a2, b2, c2), a2 < b2.
  a2 < b2 :-
    Inv(a1, b1, c1, a2, b2, c2),
    SchFT(a1, b1, c1, a2, b2, c2), a1 < b1.
  SchTF(a1, b1, c1, a2, b2, c2),
  SchFT(a1, b1, c1, a2, b2, c2),
  SchTT(a1, b1, c1, a2, b2, c2) :-
    Inv(a1, b1, c1, a2, b2, c2),
    a1 < b1 or a2 < b2.
  
  c2 < c1 :- Inv(a1, b1, c1, a2, b2, c2), a1 >= b1, a2 >= b2.
  \end{alltt}
  In the experiment, we provided the following constraint as a hint:
  \begin{alltt}
  a1 > 0, b2 < b1 :- Inv(a1, b1, c1, a2, b2, c2).
  \end{alltt}
  \item The co-termination verification problem \verb|CotermIntro| from Example~\ref{ex:coterm}.
  \item The TS-GNI verification problem \verb|TS_GNI_h**| from Example~\ref{ex:tsgni}.
  In the experiment of \verb|TS_GNI_hFT|, we provided the following constraint as a hint:
  \begin{alltt}
  x1 >= low1 :- Inv(pr, d, b, b1 : bool, x1, low1, b2 : bool, x2, low2).
  \end{alltt}
  In the experiment of \verb|TS_GNI_hTT|, we provided the following constraint as a hint:
  \begin{alltt}
  b1 or x1 >= low1 :- Inv(pr, d, b, b1 : bool, x1, low1, b2 : bool, x2, low2).
  \end{alltt}
  Note that these are a part of necessary non-relational invariant.
  \item The TS-GNI verification problem \verb|SimpleTS_GNI1| of the following program:
  \begin{alltt}
    while ( * ) \{ x ++; \}; return (high + x);
  \end{alltt}
  The encoded constraints are:
  \begin{alltt}
  Inv(0, b, b1 : bool, x1, high1, b2 : bool, x2, high2) :-
    FN_DB(x1, high1, x2, high2, b), b1, b2, x1 = x2.
  Inv(d', b, b1' : bool, x1', high1, b2 : bool, x2, high2) :-
    Inv(d, b, b1 : bool, x1, high1, b2 : bool, x2, high2),
    SchTF(d, b, b1 : bool, x1, high1, b2 : bool, x2, high2),
    b1 and (b1' and x1' = x1 + 1 or !b1' and x1' = x1) or
    !b1 and !b1' and x1' = x1,
    (!b1 or !b2 or d' = d + 1).
  Inv(d', b, b1 : bool, x1, high1, b21 : bool, x21, high2),
  Inv(d', b, b1 : bool, x1, high1, b22 : bool, x22, high2) :-
    Inv(d, b, b1 : bool, x1, high1, b2 : bool, x2, high2),
    SchFT(d, b, b1 : bool, x1, high1, b2 : bool, x2, high2),
    b2 and b21 and x21 = x2 + 1 and !b22 and x22 = x2 or
    !b2 and !b21 and x21 = x2 and !b22 and x22 = x1,
    (!b1 or !b2 or d' = d - 1).
  Inv(d, b, b1' : bool, x1', high1, b21 : bool, x21, high2),
  Inv(d, b, b1' : bool, x1', high1, b22 : bool, x22, high2) :-
    Inv(d, b, b1 : bool, x1, high1, b2 : bool, x2, high2),
    SchTT(d, b, b1 : bool, x1, high1, b2 : bool, x2, high2),
    b1 and (b1' and x1' = x1 + 1 or !b1' and x1' = x1) or
    !b1 and !b1' and x1' = x1,
    b2 and b21 and x21 = x2 + 1 and !b22 and x22 = x2 or
    !b2 and !b21 and x21 = x2 and !b22 and x22 = x1.
  
  b1 :-
    Inv(d, b, b1 : bool, x1, high1, b2 : bool, x2, high2),
    SchTF(d, b, b1 : bool, x1, high1, b2 : bool, x2, high2),
    b2.
  b2 :-
    Inv(d, b, b1 : bool, x1, high1, b2 : bool, x2, high2),
    SchFT(d, b, b1 : bool, x1, high1, b2 : bool, x2, high2),
    b1.
  SchTF(d, b, b1 : bool, x1, high1, b2 : bool, x2, high2),
  SchFT(d, b, b1 : bool, x1, high1, b2 : bool, x2, high2),
  SchTT(d, b, b1 : bool, x1, high1, b2 : bool, x2, high2) :-
    Inv(d, b, b1 : bool, x1, high1, b2 : bool, x2, high2), b1 or b2.
  -b <= d and d <= b and b >= 0 :-
    Inv(d, b, b1 : bool, x1, high1, b2 : bool, x2, high2), b1, b2.
  
  high1 + x1 = high2 + x2 :-
    Inv(d, b, b1 : bool, x1, high1, b2 : bool, x2, high2), !b1, !b2.
  
  WF_R2(b2 : bool, x2, high2, b21 : bool, x21, high2),
  WF_R2(b2 : bool, x2, high2, b22 : bool, x22, high2) :-
    Inv(d, b, b1 : bool, x1, high1, b2 : bool, x2, high2),
    !b1, b2 and b21 and x21 = x2 + 1 and !b22 and x22 = x2.
  \end{alltt}
  \item The TS-GNI verification problem \verb|SimpleTS_GNI2| of the following program:
  \begin{alltt}
    x = high; while ( * ) \{ x ++; \}; return x;
  \end{alltt}
  The encoded constraints are:
  \begin{alltt}
  Inv(0, b, b1 : bool, x1, high1, b2 : bool, x2, high2) :-
    FN_DB(x1, high1, x2, high2, b), b1, b2, x1 = high1, x2 = high2.
  Inv(d', b, b1' : bool, x1', high1, b2 : bool, x2, high2) :-
    Inv(d, b, b1 : bool, x1, high1, b2 : bool, x2, high2),
    SchTF(d, b, b1 : bool, x1, high1, b2 : bool, x2, high2),
    b1 and (b1' and x1' = x1 + 1 or !b1' and x1' = x1) or
    !b1 and !b1' and x1' = x1,
    (!b1 or !b2 or d' = d + 1).
  Inv(d', b, b1 : bool, x1, high1, b21 : bool, x21, high2),
  Inv(d', b, b1 : bool, x1, high1, b22 : bool, x22, high2) :-
    Inv(d, b, b1 : bool, x1, high1, b2 : bool, x2, high2),
    SchFT(d, b, b1 : bool, x1, high1, b2 : bool, x2, high2),
    b2 and b21 and x21 = x2 + 1 and !b22 and x22 = x2 or
    !b2 and !b21 and x21 = x2 and !b22 and x22 = x1,
    (!b1 or !b2 or d' = d - 1).
  Inv(d, b, b1' : bool, x1', high1, b21 : bool, x21, high2),
  Inv(d, b, b1' : bool, x1', high1, b22 : bool, x22, high2) :-
    Inv(d, b, b1 : bool, x1, high1, b2 : bool, x2, high2),
    SchTT(d, b, b1 : bool, x1, high1, b2 : bool, x2, high2),
    b1 and (b1' and x1' = x1 + 1 or !b1' and x1' = x1) or
    !b1 and !b1' and x1' = x1,
    b2 and b21 and x21 = x2 + 1 and !b22 and x22 = x2 or
    !b2 and !b21 and x21 = x2 and !b22 and x22 = x1.
  
  b1 :-
    Inv(d, b, b1 : bool, x1, high1, b2 : bool, x2, high2),
    SchTF(d, b, b1 : bool, x1, high1, b2 : bool, x2, high2),
    b2.
  b2 :-
    Inv(d, b, b1 : bool, x1, high1, b2 : bool, x2, high2),
    SchFT(d, b, b1 : bool, x1, high1, b2 : bool, x2, high2),
    b1.
  SchTF(d, b, b1 : bool, x1, high1, b2 : bool, x2, high2),
  SchFT(d, b, b1 : bool, x1, high1, b2 : bool, x2, high2),
  SchTT(d, b, b1 : bool, x1, high1, b2 : bool, x2, high2) :-
    Inv(d, b, b1 : bool, x1, high1, b2 : bool, x2, high2), b1 or b2.
  -b <= d and d <= b and b >= 0 :-
    Inv(d, b, b1 : bool, x1, high1, b2 : bool, x2, high2), b1, b2.
  
  x1 = x2 :- Inv(d, b, b1 : bool, x1, high1, b2 : bool, x2, high2), !b1, !b2.
  
  WF_R2(b2 : bool, x2, high2, b21 : bool, x21, high2),
  WF_R2(b2 : bool, x2, high2, b22 : bool, x22, high2) :-
    Inv(d, b, b1 : bool, x1, high1, b2 : bool, x2, high2),
    !b1, b2 and b21 and x21 = x2 + 1 and !b22 and x22 = x2.
  \end{alltt}
  \item The TS-GNI verification problem \verb|InfBranchTS_GNI| of the following program:
  \begin{alltt}
    if (high) \{
      while (x>0) \{ x = x - max( * , 1); \}
    \} else \{
      while (x>0) \{ x = x - 1; \}
    \}
  \end{alltt}
  The encoded constraints are:
  \begin{alltt}
    Inv(0, b, high1 : bool, x1, high2 : bool, x2) :-
      FN_DB(high1 : bool, x1, high2 : bool, x2, b), x1 = x2.
    Inv(d', b, high1 : bool, x1', high2 : bool, x2) :-
      Inv(d, b, high1 : bool, x1, high2 : bool, x2),
      SchTF(d, b, high1 : bool, x1, high2 : bool, x2),
      high1 and x1 > 0 and FN_R(high1 : bool, x1, nd) and
      (nd >= 1 and x1' = x1 - nd or x1' = x1 - 1) and d' = d + 1 or
      !high1 and x1 > 0 and x1' = x1 - 1 or
      x1 <= 0 and x1' = x1,
      (x1 <= 0 or x2 <= 0 or d' = d + 1).
    Inv(d', b, high1 : bool, x1, high2 : bool, x2') :-
      Inv(d, b, high1 : bool, x1, high2 : bool, x2),
      SchFT(d, b, high1 : bool, x1, high2 : bool, x2),
      high2 and x2 > 0 and nd >= 1 and x2' = x2 - nd or
      !high2 and x2 > 0 and x2' = x2 - 1 or
      x2 <= 0 and x2' = x2,
      (x1 <= 0 or x2 <= 0 or d' = d - 1).
    Inv(d, b, high1 : bool, x1', high2 : bool, x2') :-
      Inv(d, b, high1 : bool, x1, high2 : bool, x2),
      SchTT(d, b, high1 : bool, x1, high2 : bool, x2),
      high1 and x1 > 0 and FN_R(high1 : bool, x1, nd) and
      (nd >= 1 and x1' = x1 - nd or x1' = x1 - 1) or
      !high1 and x1 > 0 and x1' = x1 - 1 or
      x1 <= 0 and x1' = x1,
      high2 and x2 > 0 and nd >= 1 and x2' = x2 - nd or
      !high2 and x2 > 0 and x2' = x2 - 1 or
      x2 <= 0 and x2' = x2.
    
    x1 > 0 :-
      Inv(d, b, high1 : bool, x1, high2 : bool, x2),
      SchTF(d, b, high1 : bool, x1, high2 : bool, x2),
      x2 > 0.
    x2 > 0 :-
      Inv(d, b, high1 : bool, x1, high2 : bool, x2),
      SchFT(d, b, high1 : bool, x1, high2 : bool, x2),
      x1 > 0.
    SchTF(d, b, high1 : bool, x1, high2 : bool, x2),
    SchFT(d, b, high1 : bool, x1, high2 : bool, x2),
    SchTT(d, b, high1 : bool, x1, high2 : bool, x2) :-
      Inv(d, b, high1 : bool, x1, high2 : bool, x2), x1 > 0 or x2 > 0.
    -b <= d and d <= b and b >= 0 :-
      Inv(d, b, high1 : bool, x1, high2 : bool, x2), x1 > 0, x2 > 0.
    
    top :- Inv(d, b, high1 : bool, x1, high2 : bool, x2), x1 <= 0, x2 <= 0.
    
    WF_R1(high1 : bool, x1, high1 : bool, x1') :-
      Inv(d, b, high1 : bool, x1, high2 : bool, x2),
      x2 <= 0, 
      high1 and x1 > 0 and FN_R(high1 : bool, x1, nd) and
      (nd >= 1 and x1' = x1 - nd or x1' = x1 - 1) or
      !high1 and x1 > 0 and x1' = x1 - 1.
  \end{alltt}
  \item The TI-GNI verification problem \verb|TI_GNI_h**| of the following program:
  \begin{alltt}
    if (high) \{
      x = *;
      if (x >= low) \{ return x; \} else \{ return low; \}
    \} else \{
      x = low;
      while ( * ) \{ x++; \}
      return x;
    \}
  \end{alltt}
  The encoded constraints are:
  \begin{alltt}
Inv(pr (* prophecy variable for the return value of Copy 1 *),
    b1 : bool, x1, high1 : bool, low1, b2 : bool, x2, high2 : bool, low2) :-
  b1, b2, low1 = low2,
  high1 and x1 = nd1 or
  !high1 and x1 = low1,
  high2 and FN_R(pr, high2 : bool, low2, x2) or
  !high2 and x2 = low2.
Inv(pr, b1' : bool, x1', high1 : bool, low1, b2 : bool, x2, high2 : bool, low2) :-
  Inv(pr, b1 : bool, x1, high1 : bool, low1, b2 : bool, x2, high2 : bool, low2),
  SchTF(pr, b1 : bool, x1, high1 : bool, low1, b2 : bool, x2, high2 : bool, low2),
  high1 and b1 and (x1 >= low1 and !b1' and x1' = x1 or
                    x1 < low1 and !b1' and x1' = low1) or
  !high1 and b1 and (b1' and x1' = x1 + 1 or
                     !b1' and x1' = x1) or
  !b1 and !b1' and x1' = x1.
Inv(pr, b1 : bool, x1, high1 : bool, low1, b21 : bool, x21, high2 : bool, low2),
Inv(pr, b1 : bool, x1, high1 : bool, low1, b22 : bool, x22, high2 : bool, low2) :-
  Inv(pr, b1 : bool, x1, high1 : bool, low1, b2 : bool, x2, high2 : bool, low2),
  SchFT(pr, b1 : bool, x1, high1 : bool, low1, b2 : bool, x2, high2 : bool, low2),
  high2 and b2 and (x2 >= low2 and !b21 and x21 = x2 and
                                   !b22 and x22 = x2 or
                    x2 < low2 and !b21 and x21 = low2 and
                                  !b22 and x22 = low2) or
  !high2 and b2 and b21 and x21 = x2 + 1 and
                    !b22 and x22 = x2 or
  !b2 and !b21 and x21 = x2 and
          !b22 and x22 = x2.
Inv(pr, b1' : bool, x1', high1 : bool, low1, b21 : bool, x21, high2 : bool, low2),
Inv(pr, b1' : bool, x1', high1 : bool, low1, b22 : bool, x22, high2 : bool, low2) :-
  Inv(pr, b1 : bool, x1, high1 : bool, low1, b2 : bool, x2, high2 : bool, low2),
  SchTT(pr, b1 : bool, x1, high1 : bool, low1, b2 : bool, x2, high2 : bool, low2),
  high1 and b1 and (x1 >= low1 and !b1' and x1' = x1 or
                    x1 < low1 and !b1' and x1' = low1) or
  !high1 and b1 and (b1' and x1' = x1 + 1 or
                     !b1' and x1' = x1) or
  !b1 and !b1' and x1' = x1,
  high2 and b2 and (x2 >= low2 and !b21 and x21 = x2 and
                                   !b22 and x22 = x2 or
                    x2 < low2 and !b21 and x21 = low2 and
                                  !b22 and x22 = low2) or
  !high2 and b2 and b21 and x21 = x2 + 1 and
                    !b22 and x22 = x2 or
  !b2 and !b21 and x21 = x2 and
          !b22 and x22 = x2.

b1 or pr <> x1 :-
  Inv(pr, b1 : bool, x1, high1 : bool, low1, b2 : bool, x2, high2 : bool, low2),
  SchTF(pr, b1 : bool, x1, high1 : bool, low1, b2 : bool, x2, high2 : bool, low2),
  b2.
b2 :-
  Inv(pr, b1 : bool, x1, high1 : bool, low1, b2 : bool, x2, high2 : bool, low2),
  SchFT(pr, b1 : bool, x1, high1 : bool, low1, b2 : bool, x2, high2 : bool, low2),
  b1 or pr <> x1.
SchTF(pr, b1 : bool, x1, high1 : bool, low1, b2 : bool, x2, high2 : bool, low2),
SchFT(pr, b1 : bool, x1, high1 : bool, low1, b2 : bool, x2, high2 : bool, low2),
SchTT(pr, b1 : bool, x1, high1 : bool, low1, b2 : bool, x2, high2 : bool, low2) :-
  Inv(pr, b1 : bool, x1, high1 : bool, low1, b2 : bool, x2, high2 : bool, low2),
  b1 or pr <> x1 or b2.

x1 = x2 :-
  Inv(pr, b1 : bool, x1, high1 : bool, low1, b2 : bool, x2, high2 : bool, low2),
  !b1 and pr = x1 (* if the prophecy is correct *), !b2.
  \end{alltt}
  In the experiment of \verb|TI_GNI_hFT|, we provided the following constraint as a hint:
  \begin{alltt}
  x1 >= low1 :- Inv(pr, b1 : bool, x1, low1, b2 : bool, x2, low2).
  \end{alltt}
  Note that this is a part of necessary non-relational invariant.
\end{itemize}

}{}

\end{document}